\documentclass[conference]{IEEEtran}

\makeatletter
\newcommand\semihuge{\@setfontsize\semihuge{22.3}{22.6}}
\makeatother

\usepackage{algpseudocode}
\usepackage{algorithm}
\usepackage{algorithmicx}

\usepackage{lipsum} 

\usepackage[dvips]{color}
\usepackage{comment}
\usepackage{todonotes}
\usepackage{epsf}
\usepackage{epsfig}
\usepackage{times}
\usepackage{epsfig}
\usepackage{graphicx}
\usepackage{bbold}
\usepackage{mathtools}
\usepackage{mathrsfs}
\usepackage{amssymb}
\usepackage{pdfpages}
\usepackage{epstopdf}
\usepackage{float}
\newfloat{algorithm}{t}{lop}
\usepackage{dsfont}
\usepackage{lettrine} 
\usepackage{amsmath,epsfig,amssymb,algorithm,algpseudocode,amsthm,cite,url}
\usepackage{subcaption}
\allowdisplaybreaks
\usepackage{csquotes}

\usepackage{verbatim}
\usepackage[english]{babel}
\usepackage{amsmath,amssymb}

\usepackage[justification=centering]{caption}
\usepackage{verbatim}

\newtheorem{theorem}{\bf Theorem}

\newtheorem{lemma}{\bf Lemma}

\IEEEoverridecommandlockouts

\begin{document}
\title{ \Huge Communications and Control for Wireless Drone-Based Antenna Array\vspace{-0.001cm}}    

\author{\IEEEauthorblockN{  Mohammad Mozaffari$^1$, Walid Saad$^2$, Mehdi Bennis$^3$, and M\'erouane Debbah$^4$}\vspace{-0.1cm}\\
	\IEEEauthorblockA{
		\small $^1$ Ericsson, Santa Clara, CA, USA, Email: \url{mohammad.mozaffari@ericsson.com}.\vspace{-0.00cm}\\
		$^2$ Wireless@VT, Electrical and Computer Engineering Department, Virginia Tech, VA, USA, Email:\url{walids@vt.edu}.\vspace{-0.00cm}\\
		$^3$ CWC - Centre for Wireless Communications, University of Oulu, Finland, Email: \url{bennis@ee.oulu.fi}.\\ 
		$^4$ Mathematical and Algorithmic Sciences Lab, Huawei France R\&D, Paris, France, and CentraleSup´elec,\\   Universit\'e Paris-Saclay, Gif-sur-Yvette, France, Email: \url{merouane.debbah@huawei.com}.
		\thanks{Mohammad Mozaffari joined Ericsson in July 2018. He was with Wireless@VT, Electrical and Computer Engineering Department, Virginia Tech, VA, USA, when this work was done.}
	}\vspace{-0.6cm}}
\maketitle

\begin{abstract}
	In this paper, the effective use of multiple quadrotor drones as an aerial antenna array that provides wireless service to ground users is investigated. In particular, under the goal of  minimizing the \emph{airborne service time} needed for communicating with ground users, a novel framework for deploying and operating a drone-based antenna array system whose  elements are single-antenna drones is proposed. In the considered model, the service time is minimized by minimizing the wireless transmission time as well as the control time that is needed for movement and stabilization of the drones. 
	 To minimize the transmission time, first, the antenna array gain is maximized by optimizing the drone spacing within the array. In this case, using perturbation techniques, the drone spacing optimization problem is addressed by solving successive, perturbed convex optimization problems. Then, according to the location of each ground user, the optimal locations of the drones around the array's center are derived such that the transmission time for the user is minimized. Given the determined optimal locations of drones, the drones must spend a control time to adjust their positions dynamically so as to serve multiple users. To minimize this control time of the quadrotor drones, the speed of rotors is optimally adjusted based on both the destinations of the drones and external forces (e.g., wind and gravity). In particular, using \emph{bang-bang} control theory, the optimal rotors' speeds as well as the minimum control time are derived in closed-form.  \textcolor{black}{Simulation results show that the proposed approach can significantly reduce the service time to ground users compared to a fixed-array case  in which the same number of drones form a fixed uniform antenna array. The results also show that, in comparison with the fixed-array case, the network's spectral efficiency can be improved by 32\% while leveraging the drone antenna array system. Finally, the results reveal an inherent tradeoff between the control time and transmission time while varying the number of drones in the array.}  

\end{abstract} 

\section{Introduction}
The use of unmanned aerial vehicles (UAVs) such as drones is growing rapidly
across many domains including delivery, communications, 
surveillance, and
search and rescue in emergency operations \cite{ mozaffari2018tutorial,Qin, alzenad, mozaffari2,TrajectoryZhang, wu2018uav}. In wireless networks, drones
can be used as flying base stations to provide reliable and cost-effective wireless connectivity \cite{mozaffari2,alzenad,TrajectoryZhang, Qin, Jeong, Letter, VshalUAV, LetterSudheesh,  bor, wu2018uav,mozaffari2018beyond}. Due to their flexibility, agility, and mobility, drones can support reliable, cost-effective, and high data rate wireless
communications for ground users. In particular,
during major public events such as Olympic games that generate a substantial demand for communication, there is a need to supplement the limited capacity and coverage capabilities of existing cellular networking infrastructure. In such scenarios, drone-based wireless communication is an ideal solution. For instance, AT\&T and Verizon are planning to use flying drones to boost the Internet coverage for
the college football national championship and the Super Bowl.
Drones can also play a key role in enabling wireless connectivity in other key scenarios such as public safety, and Internet of Things (IoT) scenarios \cite{mozaffari2}. To effectively leverage drones for wireless networking applications, one must address a number of challenges that include optimal placement of drones, path planning, resource management, control, and flight time optimization \cite{mozaffari2, Qin, bor}.

\subsection{Related work on UAV communications} 
There has been a recent surge of literature discussing the use of drones for wireless communication purposes \cite{mozaffari2,alzenad,Qin, Azari, bor, VshalUAV, Lyu, MozaffariFlightTime, Complition,TrajectoryZhang, wu2018uav,Jeong}. For instance, in \cite{alzenad}, the authors studied the optimal 3D placement of UAVs for maximizing the number of covered users  
with different quality-of-service (QoS) requirements. The works in \cite{Qin} and \cite{Jeong} studied path planning and optimal deployment problems for UAV-based communications and computing. 
The work in \cite{VshalUAV} proposed a framework for the optimal placement and distribution of UAVs to minimize the overall delay in a UAV-assisted wireless network.  A comparison between the performance of aerial base stations and terrestrial base stations in terms of average sum rate and transmit power is presented in \cite{Azari}. In \cite{Lyu}, a polynomial-time algorithm for the optimal placement of drones that provide coverage for ground terminals is proposed.

One of the fundamental challenges in drone-based communications systems is the limited flight endurance of drones. Naturally, flying drones have a limited amount of on-board energy which must be used for transmission, mobility, control,
data processing, and payloads purposes. Consequently, the flight duration of
drones is typically short and can be insufficient for providing a long-term, continuous wireless coverage. Furthermore, due to the limited transmit power of drones, providing long-range, high rate, and low latency communications can be challenging in drone-enabled wireless systems. In this regard, a key performance metric in drone-enabled wireless networks is \emph{airborne service time}, which is defined as the time needed for servicing ground users. The service time directly impacts the flight time of drones as well as the quality-of-service (i.e., delay) for ground users.  From the drones' perspective, a
lower service time corresponds to a shorter flight time as well as less energy consumption. From the users' point of view, a lower service time is also needed as it directly yields lower latency.  \textcolor{black}{To address the flight time and energy consumption challenges of drones, the authors  in \cite{TrajectoryZhang} proposed a comprehensive analytical framework for optimizing the trajectory of a fixed-wing UAV with the objective of minimizing the UAV's energy consumption while serving a ground user. In particular, a new design paradigm is developed that jointly considers the communication rate and the UAV's energy consumption.} The work in \cite{MozaffariFlightTime} minimized the hover time of drone base stations by deriving the optimal cell association schemes. However, the model in \cite{MozaffariFlightTime} is limited to static single-antenna drones.  
In \cite{Complition}, the trajectory and mission completion time of a single UAV that serves ground users are optimized. However, the work in \cite{Complition} does not analyze a scenario with multiple UAVs.


One promising approach to provide high data rate and low service time is to utilize multiple drones within an antenna array system composed of multiple single-antenna drones \cite{Garza}.
 Compared to conventional antenna array systems, a drone-based antenna array has the following advantages. First, the number of antenna elements (i.e., drones) is not limited by space constraints. Second, the gain of the drone-based antenna array can be increased by adjusting the array element spacing. Third, the mobility and flexibility of drones enable an effective mechanical beam-steering in any three-dimensional (3D) direction. Clearly, a high gain drone-based antenna array can provide high data rate wireless services to ground users thus reducing the service time.      


In \cite{Garza}, the authors studied the design of a UAV-based antenna array for directivity maximization. However, the approach presented in \cite{Garza} is based on a heuristic and a computationally demanding evolutionary algorithm. Moreover, the service time analysis is ignored in \cite{Garza}. In \cite{Weif}, the authors derived the asymptotic capacity of an airborne multiple-input-multiple-output (MIMO)  wireless communication system. However, the work in \cite{Weif} considers fixed positions for the antenna elements of the transmitter and the receiver. Furthermore, this work does not analyze the control aspect of drones which is essential in designing drone-based MIMO systems. In fact, none of the previous works on drone communications, such as in \cite{mozaffari2, alzenad,Qin,Azari, VshalUAV, Jeong, Lyu, bor,Complition,Garza,Weif,MozaffariFlightTime,wu2018uav,mozaffari2018beyond,Naderi,LetterSudheesh,Letter}, has studied the use of a drone-based antenna array system for service time minimization.


We note that, there exist some studies on time-optimal motion planning \cite{manipulators, Robust, computationally, TimeOptimalDrone}. However, most of the previous works do not address the time-optimal control problem of quadrotor drones. While the authors in \cite{TimeOptimalDrone} consider a quadrotor drone in their model, they ignore the effect of external forces on the control time. Furthermore, the approach in \cite{TimeOptimalDrone} is based on a genetic algorithm which is computationally demanding. Unlike our work, the work in \cite{TimeOptimalDrone} ignores the communication aspects of drones, and does not capture the impact of control time on the performance of drone-enabled wireless networks. Compared to \cite{TimeOptimalDrone}, our proposed framework comprises both communication and control aspects of drones and it is analytically tractable. 
\subsection{Contributions}
The main contribution of this paper is a novel framework for deploying and operating a drone-based antenna array system that delivers wireless service to a number of ground users within a minimum time. In particular, we minimize the service time that includes both the transmission time and the control time needed to control the movement and orientation of the drones. To this end, we minimize the transmission time, by optimizing the drones' locations, as well as the control time that the drones need to move between these optimal locations. To minimize the transmission time,  first, we determine the optimal drone spacing for which the array directivity is maximized. In this case, using perturbation theory \cite{OptimizationPertub}, we solve the drone spacing optimization problem by successively solving a number of perturbed convex optimization problems. Next, given the derived drone spacing, we optimally adjust the locations of the drones according to the position of each ground user. In order to serve different users, the drones must dynamically move between the derived optimal locations, during the control time period. To minimize the control time of quadrotor drones, we determine the optimal speeds of rotors such that the drones can update their positions and orientations within a minimum time. In this case, using \emph{bang-bang} control theory \cite{IntroControl}, we derive a closed-form expression for the minimum control time as a function of external forces (e.g., wind and gravity), the drone's weight, and the destinations of drones. \textcolor{black}{Our results show that the proposed drone antenna array approach can significantly reduce the service time and improve the spectral and energy efficiency of the network. In particular, our approach yields 32\% improvement in spectral efficiency compared to a case in which the same number of drones form a fixed uniform aerial antenna array. The results also reveal a tradeoff between the control time and transmission time while varying the number of drones.}

\section{System Model and General Problem Formulation}

Consider a set $\mathcal{L}$ of $L$ single-antenna wireless users located within a given geographical area. In this area, a set  $\mathcal{M}$ of $M$ quadrotor drones are used as flying access points to provide downlink wireless service for ground users. The $M$ drones will form an antenna array in which each element is a single-antenna drone, as shown in Fig.\,\ref{SystemModel}. For tractability, we consider a linear antenna array whose elements are symmetrically excited and located about the origin of the array as done in \cite{Purturbation}.  
The results that we will derive for the linear array case can provide a key guideline for designing more complex 2D and 3D array configurations. The 3D location of drone $m\in \mathcal{M}$ and of user $i\in \mathcal{L}$  is given by ($x^\textrm{u} _{i},y^\textrm{u} _{i},z^\textrm{u}_i)$, and the location of drone $m$ while serving user $i$ is $(x_{m,i},y_{m,i},z_{m,i})$. 
To avoid collisions, we assume that adjacent drones in the array are separated by at least  $D_\textrm{min}$. Let $a_m$ and $\beta_m$ be the amplitude and phase of the signal (i.e. excitation) at element $m$ in the array. 
Let \textcolor{black}{${d_{m,i}} = \sqrt {{{(\textcolor{black}{x_{m,i}} - {x_{o}})}^2} + {{({y_{m,i}} - {y_{o}})}^2} + {{({z_{m,i}} - {z_{o}})}^2}}$} be the distance of drone $m$ from the origin of the array whose 3D coordinate is $(x_o,y_o,z_o)$.   
The magnitude of the far-field radiation pattern of each element is $w{(\theta,\phi)}$, where $\theta$ and $\phi$ are the polar and azimuthal angles in the spherical coordinate.

\begin{figure}[!t]
	\begin{center}
		\vspace{-0.1cm}
		\includegraphics[width=8.5 cm]{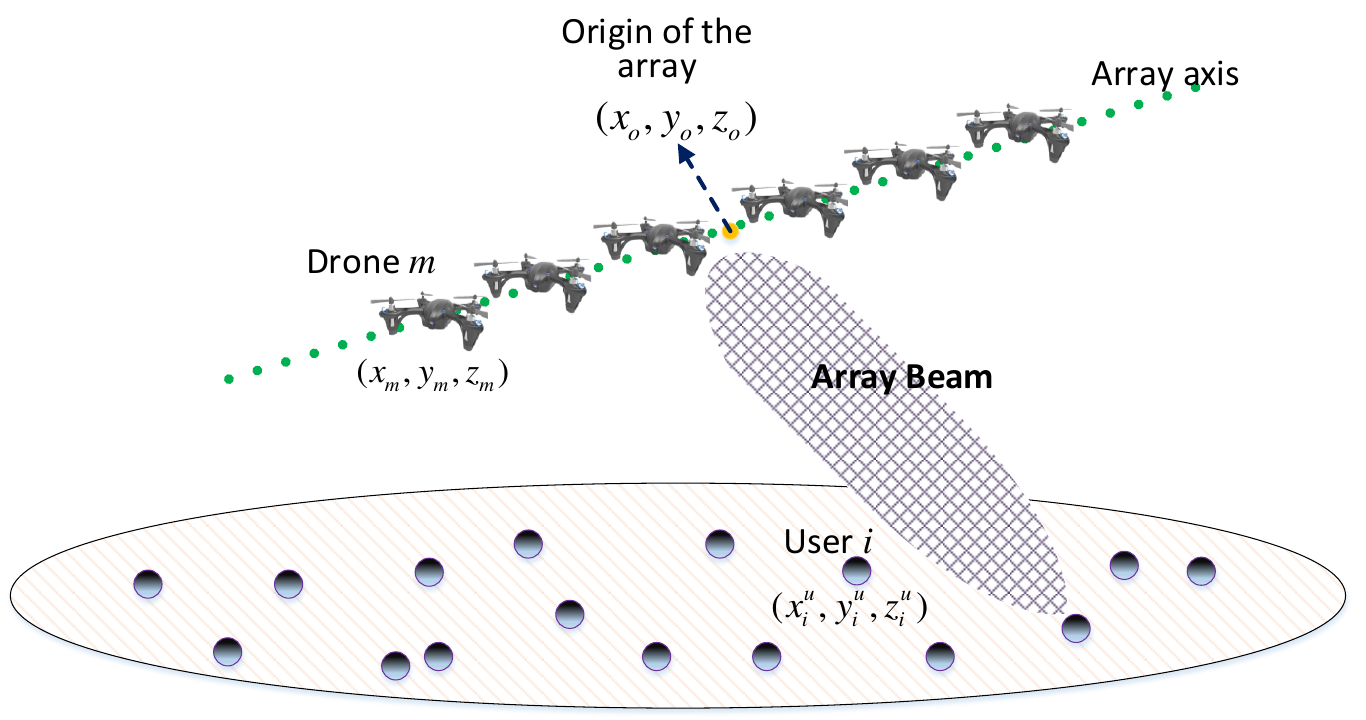}
		\vspace{-0.1cm}
		\caption{  \small Drone-based antenna array.}\vspace{-0.3cm}
		\label{SystemModel}
	\end{center}	
\end{figure}

To serve ground users distributed over a geographical area, the drones will dynamically change their positions based on each user's location.  \textcolor{black}{In our model, drones hover at specific locations to serve a user, and fly to a new position to serve
another user.} Such repositioning is needed for adjusting the distance and beam direction of the antenna array to each ground user. \textcolor{black}{We consider a ``fly-then-hover-and-transmit" operation (as also done in \cite{HoverZhang}) for the drone-based antenna array system. In this case, drones transmit when they are stationary and, hence, transmission is not performed while the array moves. Such a transmission protocol is suitable for the considered drone-based antenna array system since the antenna array needs to be stable so as to effectively perform beamforming and to establish reliable communication links to ground users.} Note that, unlike a classical linear phased array that uses electronic beam steering, the proposed drone-based antenna array relies on the repositioning of drones{\footnote{\textcolor{black}{In general, the array gain depends on the elements' positions and the phase of the elements. In classical antenna array systems with fixed elements, the phase of the elements is often optimized. Here, we exploit the drones' flexibility to maximize the array directivity by optimizing the element (i.e., drone) spacing, given the elements' phases.}}}. This is due to the fact that, in the drone antenna array, precisely adjusting the elements' phase is more challenging than the phased array whose elements are directly connected. In addition, a linear phased array cannot perform 3D beam steering. Hence, in our model, the drones dynamically adjust their positions in order to steer the beam towards ground users.    
Clearly, the \emph{service time}, which is the time needed to serve the ground users, depends on the transmission time  and the control time during which the drones must move and stabilize their locations. The transmission time is inversely proportional to the downlink data rate which depends on the signal-to-noise-ratio (SNR) which is, in turn, function of the array's beamforming gain. 

The service time is an important metric for both users and drones. A lower service time yields a lower delay and, hence, higher quality-of-service for the users. Also, the service time is directly related to spectral efficiency as it depends on data rate and transmission bandwidth. For drones, a lower service time corresponds to a shorter flight time and less energy consumption. 
In fact, minimizing the service time improves both energy and spectral efficiency. Therefore, our goal is to minimize the total service time of the ground users by optimally adjusting the drones' locations, within a minimum control time, that can provide a maximum data rate.\\   
\indent For drone-to-ground communications, we consider a line-of-sight (LoS) propagation model as  done in \cite{Qin} and \cite{Complition}. Such a channel model is reasonable here as the effect of multipath is significantly mitigated due to the high altitude of drones and using beamforming \cite{Complition}. The transmission rate from the drone antenna array to ground user $i$ in a far-field region is given by \cite{Complition}: 
\begin{equation}
	R_i(\boldsymbol{x}_i,\boldsymbol{y}_i,\boldsymbol{z}_i)={B{{\log }_2}\left( {1 + \frac{{r_i^{ - \alpha }{P_t}K_o{G_i}(\boldsymbol{x}_i,\boldsymbol{y}_i,\boldsymbol{z}_i)}}{{{\sigma ^2}}}} \right)}, \label{Rate}
\end{equation}
where $\boldsymbol{x}_i=[x_{m,i}]_{M\times1}$, $\boldsymbol{y}_i=[y_{m,i}]_{M\times1}$, $\boldsymbol{z}_i=[z_{m,i}]_{M\times1}$, $m\in \mathcal{M}$ representing the 3D coordinates of the drones while serving user $i$.  
  $B$ is the transmission bandwidth, $r_i$ is the distance between the origin of the array and user $i$, $P_t$ is the total transmit power of the array, $\sigma^2$ is the noise power, and $K_o$ is the constant path loss coefficient. ${G_i}(\boldsymbol{x}_i,\boldsymbol{y}_i,\boldsymbol{z}_i)$ is the gain of the antenna array towards the location of user $i$. \textcolor{black}{In the proposed drone-based antenna array system, each drone is an antenna element of the array. In this case, the entire antenna array can be modeled as a single directional antenna whose gain is the total array gain \cite{Rabert1}.} The array gain is  given by \cite{AntennaDesign}: 
\begin{equation}
	{G_i}(\boldsymbol{x}_i,\boldsymbol{y}_i,\boldsymbol{z}_i) = \frac{{4\pi {{\left| {F({\theta _i},{\phi _i})} \right|}^2}w{{({\theta _i},{\phi _i})}^2}}}{{\int\limits_0^{2\pi } {\int\limits_0^\pi  {{{\left| {F(\theta ,\phi )} \right|}^2}w{{(\theta ,\phi )}^2}\sin \theta \textrm{d}\theta \textrm{d}\phi } } }}\eta, \label{gain}
\end{equation}
where $0\le \eta \le 1$ is the antenna array efficiency which is multiplied by directivity to compute the antenna gain. In fact, the antenna gain is equal to the antenna directivity multiplied by $\eta$. In (\ref{gain}), $F(\theta , \phi)$ is the array factor which can be written as \cite{AntennaDesign}:
\begin{equation}
F(\theta ,\phi ) \hspace{-0.1cm}= \hspace{-0.1cm}\sum\limits_{m = 1}^M {\hspace{-0.1cm}{a_m}{e^{j\left[ {k\left( {{x_{m,i}}\sin \theta \cos \phi  + {y_{m,i}}\sin \theta \sin \phi  + {z_{m,i}}\cos \theta } \right) + {\beta _m}} \right]}}}, \label{arrayFactor}
\end{equation}
where $k=2\pi/\lambda$ is the phase constant, and $\lambda$ is the wavelength. 
Note that, the overall radiation pattern of the antenna array is equal to $F(\theta ,\phi) w({\theta _i},{\phi _i})$ which follows from the pattern multiplication rule \cite{AntennaDesign}.

Now, the total time that the drones spend to service the ground users will be:
\begin{equation} \label{Ttot}
{\textcolor{black}{T_\textrm{service}}} = \sum\limits_{i = 1}^L {\frac{{{q_i}}}{{{R_i}(\boldsymbol{x}_i,\boldsymbol{y}_i,\boldsymbol{z}_i)}} + T_i^\textrm{crl}({\boldsymbol{V}},\boldsymbol{x}_i,\boldsymbol{y}_i,\boldsymbol{z}_i)}, 
\end{equation}
where \textcolor{black}{$T_\textrm{service}$} represents the total service time, $q_i$ is the load of user $i$  which represents the number of bits that must be transmitted to user $i$. $T_i^\textrm{crl}$ is the control time during which the drones adjust their locations according to the location of ground user $i$. In particular,  $T_i^\textrm{crl}$ captures the time needed for updating the drones' locations from state $i-1$ (i.e., locations of drones while serving user $i-1$, $i>1$) to state $i$. The control time is obtained based on the dynamics of the drones and is a function of control inputs, external forces, and the movement of drones. In fact, each drone needs a vector of control inputs in order to move from its initial location to a new location while serving different users. For quadrotor drones, the rotors' speeds are commonly considered as control inputs. Therefore,  in (\ref{Ttot}), we have $\boldsymbol{V}=[v_{mn}(t)]_{M\times4}$ with $v_{mn}(t)$ being  the speed of rotor $n$ of drone $m$ at time $t$. The maximum speed of each rotor is $v_\textrm{max}$.  In this case, one can minimize the control time of the drones by properly adjusting the rotors' speeds.  In Section \ref{TimeControlSec}, we will provide a detailed analysis of the control time given the drones' dynamics.

\textcolor{black}{Clearly, to effectively employ drones within an aerial antenna array, it is crucial to ensure the stability of the drones. Hence, in the proposed drone-based antenna array system,  we adopt quadrotor drones which can hover (remain stationary) and move to any direction \cite{zhang}. In Section IV, we analyze the stability of the drones in the array when serving ground users. We derive the optimal rotors' speeds for which the quadrotor drones can stabilize their positions. Moreover, we account for wind effects while analyzing the drones' stability\footnote{\textcolor{black}{We also note that the proposed drone-based antenna array system is more suitable for a low frequency (e.g., below 600 MHz) case in which the wavelength is above 0.5 m. In this case, the array performance will not be significantly affected by drones' vibrations.}}.}

Given this model, our goal is to minimize the total service time of drones by finding the optimal locations of the drones with respect to the center of the array, as well as the optimal control inputs.  Our optimization problem, in its general form, is given by: \vspace{-0.2cm}
\begin{align} 
 &\mathop {\textrm{minimize} }\limits_{\boldsymbol{X},\boldsymbol{Y},\boldsymbol{Z},\boldsymbol{V}} \sum\limits_{i = 1}^L {\frac{{{q_i}}}{{{R_i}(\boldsymbol{x}_i,\boldsymbol{y}_i,\boldsymbol{z}_i)}} + T_i^\textrm{crl}({\boldsymbol{V}},\boldsymbol{x}_i,\boldsymbol{y}_i,\boldsymbol{z}_i)}, \label{OPT1-1}\\
 \textrm{st.}\,\, &d_{{m+1},i}-d_{m,i}\ge D_\textrm{min},\,\, \forall m \in \mathcal{M}\backslash{\{M\}}, \label{Dmin}\\
  &\textcolor{black}{0\le v_{mw}(t)\le v_\textrm{max},\,\, \forall m\in \mathcal{M}, w \in \{1,...,4\},} \label{SpeedLim}
 \end{align}
 where $\boldsymbol{X}$, $\boldsymbol{Y}$, and $\boldsymbol{Z}$ are matrices whose rows $i$ are, respectively, vectors $\boldsymbol{x}_i$, $\boldsymbol{y}_i$, and $\boldsymbol{z}_i$, $\forall i\in\mathcal{L}$. The constraint in (\ref{Dmin}) indicates that the minimum separation distance between two adjacent drones must be greater than $D_\textrm{min}$ to avoid collision. (\ref{SpeedLim}) represents the constraints on the speed of each rotor. Note that, the first term in (\ref {OPT1-1}) represents the transmission time which depends on the drones' locations. The second term, $T_i^\textrm{crl}$, is the control time which is a function of the rotors' speeds as well as the drones' locations.   
Solving (\ref{OPT1-1}) is challenging as it is highly nonlinear due to (\ref{gain}). Moreover, as we can see from (\ref{arrayFactor}), the array factor is a complex function of the array element's positions. In addition, due to the nonlinear nature of quadrotor's dynamic system, finding the optimal control inputs is a challenging task, as will be discussed in Section \ref{TimeControlSec}.  
 
We note that, considering a narrow-beam antenna array communication, (\ref{OPT1-1}) can be solved by separately optimizing drones' locations and rotors' speeds.  In the narrow-beam case, the drone array must perfectly steer its beam towards each ground user. Hence, we can first determine the optimal drones' positions and, then, optimize the rotors' speeds to move to these optimal positions within a minimum time. Our approach for solving (\ref{OPT1-1}) includes two key steps. First, given the location of any ground user, we optimize the locations of the drones in the linear array to minimize the transmission time. Thus, given $L$ ground users, we will have  $L$ sets of drones' locations.  In the second step, using the result of the first step, we determine the drones' optimal control strategy to update their locations within a minimum time. Hence, the solution of the transmission time optimization  problem (in the first step) is used as inputs to the time-optimal control problem (in the second step). While, in general, this approach leads to a suboptimal solution, it is analytically tractable and practically easy to implement.  Next, we will optimize the location of drones to achieve a minimum transmission time for any arbitrary ground user.




\section{Optimal Positions of Drones in Array for Transmission Time Minimization } \label{SecDirectivity}
In this section, we determine the optimal positions of the drones in the array based on the location of each user such that the transmission time to the user is minimized. Clearly, given (\ref{Rate}), (\ref{gain}), and (\ref{Ttot}), to minimize the transmission time, we need to maximize the array gain (i.e., directivity) towards each ground user.

Without loss of generality, we consider an even number of drones. For an odd number of drones, the same analysis will still hold.  Now, the array factor for $M$ drones located on the $x$-axis of the Cartesian coordinate can be given by:
\begin{align}
&F(\theta ,\phi ) = \sum\limits_{m = 1}^M {{a_m}{e^{j\left[ {k{x_{m,i}}\sin \theta \cos \phi  + {\beta _m}} \right]}}} \nonumber\\
 & \stackrel{(a)}{=}\sum\limits_{n = 1}^{M/2} {{a_n}\left( {{e^{j\left[ {k{d_n}\sin \theta \cos \phi  + {\beta _n}} \right]}} + {e^{ - j\left[ {k{d_n}\sin \theta \cos \phi  + {\beta _n}} \right]}}} \right)} \nonumber\\
&\stackrel{(b)}{=} 2\sum\limits_{n = 1}^N {{a_n}\cos \left( {k{d_n}\sin \theta \cos \phi  + {\beta _n}} \right)},\label{MN} 
\end{align}
 where  $N=M/2$, and $d_n$ is the distance of element $n \in \mathcal{N}=\{1,2,...,N\}$ from the center of the array (origin). Also, $(a)$ follows from the fact that the array is symmetric with respect to the origin, and $(b)$ is based on the Euler's rule.

Now, we can maximize the directivity of the array by optimizing $d_n$, $\forall n \in \mathcal{N}$:  
\begin{align} 
&\mathop {\textrm{maximize} }\limits_{{d_n},\forall n \in \mathcal{N}}  \frac{{4\pi {{\left| {F({\theta_{\max}},{\phi_{\max}})} \right|}^2}w{{({\theta_{\max}},{\phi_{\max}})}^2}}}{{\int\limits_0^{2\pi } {\int\limits_0^\pi  {{{\left| {F(\theta ,\phi )} \right|}^2}w{{(\theta ,\phi )}^2}\sin \theta \textrm{d}\theta \textrm{d}\phi } } }} \label{d1}, 
\end{align}
where $({\theta_{\max}},{\phi_{\max}})$ are the polar and azimuthal angles at which the total antenna pattern $F(\theta,\phi)w(\theta,\phi)$ has a maximum value. 
Clearly, solving (\ref{d1}) is challenging due to the non-linearity and complex expression of the objective function of this optimization problem. Moreover, this problem is non-convex and, hence, cannot be exactly solved using classical convex optimization methods. Next, we solve (\ref{d1}) by exploiting the perturbation technique \cite{Purturbation}. In general, perturbation theory aims at finding the solution of a complex problem, by starting from the exact solution of a simplified version of the original problem \cite{OptimizationPertub}. This technique is thus useful when dealing with nonlinear and analytically intractable optimization problems such as (\ref{d1}). \vspace{-0.1cm} 

\subsection{Perturbation Technique for Drone Spacing Optimization} \label{spaceP}

To optimize the distance between drones, we first consider an initial value for the distance of each drone from the origin. Then, we find the optimal perturbation value that must be added to this initial value. Let $d^0_n$ be the initial distance for drone $n$, the perturbed distance is:\vspace{-0.3cm} 
\begin{equation}
{d_n} = d_n^0 + {e_n}, \label{dn}
\end{equation}
where $e_n<<\lambda$, with $\lambda$ being the wavelength, is the perturbation value.
Given (\ref{dn}), the array factor can be approximated by:\vspace{-0.2cm}
\begin{align}
&F(\theta ,\phi ) = 2\sum\limits_{n = 1}^N {{a_n}\cos \left( {k(d_n^0 + {e_n})\sin \theta \cos \phi  + {\beta _n}} \right)} \nonumber\\
&= 2\sum\limits_{n = 1}^N {{a_n}\cos \left[ {\left( {kd_n^0\sin \theta \cos \phi  + {\beta _n}} \right) + k{e_n}\sin \theta \cos \phi } \right]}\nonumber \\
&\mathop  \approx \limits^{(a)} \sum\limits_{n = 1}^N {2{a_n}\cos \left( {kd_n^0\sin \theta \cos \phi  + {\beta _n}} \right)}\nonumber \\
&\,\,\,\,\,- \sum\limits_{n = 1}^N {2{a_n}k{e_n}\sin \theta \cos \phi \sin \left( {kd_n^0\sin \theta \cos \phi  + {\beta _n}} \right)}, \label{approx}
\end{align}
where in ($a$) we used the trigonometric properties, and the fact that $\textrm{sin}(x) \approx x$ for small values of $x$.  
Clearly, given $e_n<<\lambda$, the numerator of  (\ref{d1}) can be computed based on the values of $d_n^0$, $\forall n \in \mathcal{N}$. Hence, given $d_n^0$, our optimization problem in (\ref{d1}) can be written as:
\begin{align} 
&\mathop {\textrm{min} }\limits_{{\boldsymbol{e}}}  {{\int\limits_0^{2\pi } {\int\limits_0^\pi  {{{ {F(\theta ,\phi )} }^2}w{{(\theta ,\phi )}^2}\sin \theta \textrm{d}\theta \textrm{d}\phi } } }} \label{d3},\\%
\textrm{s.t.}\,\,\, &{d_{n+1}^0+e_{n+1}-d_n^0-e_n}\ge D_\textrm{min}, \,\, \forall n \in \mathcal{N}\,\backslash{\{N\}}, \label{emin}
\end{align} 
where $\boldsymbol{e}$ is the perturbation vector having elements $e_n$, $n \in \mathcal{N}$.  

For brevity, we define the following functions:\vspace{-0.10cm}
\begin{align}
&{F^0}(\theta ,\phi ) = \sum\limits_{n = 1}^N {{a_n}\cos \left( {kd_n^0\sin \theta \cos \phi  + {\beta _n}} \right)}, \\
&{I_{{\mathop{\rm int}} }}(x) = \int\limits_0^{2\pi } {\int\limits_0^\pi  {x\sin \theta \textrm{d}\theta \textrm{d}\phi } }.\label{IintF}
\end{align}

\begin{theorem} \label{Theor1}
	\normalfont
	The optimization problem in (\ref{d3}) is convex, and the optimal perturbation vector is the solution of the following system of equations:
	\begin{equation}
	\begin{cases}
	&\hspace{-0.3cm}\boldsymbol{e} = {\boldsymbol{G}^{ - 1}}[\boldsymbol{q}+\boldsymbol{\mu_\mathcal{L}}], \label{e_opt}\\
	&\hspace{-0.3cm} \mu_n\left( {{e_n} - {e_{n + 1}} + {D_{\textrm{min} }} + d_n^0 - d_{n + 1}^0} \right)=0, \,\, \forall n \in \mathcal{N}\,\backslash{\{N\}},\\
	&\hspace{-0.3cm} \mu_n \ge 0, \,\, \forall n \in \mathcal{N}\,\backslash{\{N\}}.
	\end{cases}
	\end{equation}
	where $\boldsymbol{G}=[g_{m,n}]_{N \times N}$ is an $N \times N$ matrix with:\vspace{-0.2cm}
	\begin{align}
	&{g_{m,n}} = {I_{{\mathop{\rm int}} }}\biggl( {a_m}{a_n}{{\left( {k\sin \theta \cos \phi w(\theta ,\phi )} \right)}^2} \nonumber\\
	&\times\sin \left( {kd_n^0\sin \theta \cos \phi  + {\beta _n}} \right)  \sin \left( {kd_m^0\sin \theta \cos \phi  + {\beta _m}} \right) \biggr), \label{G}
	\end{align}\vspace{-0.02cm}
	and $\boldsymbol{q}=[q_n]_{N\times 1}$ whose elements are given by:
	\begin{align}
	{q_n} = {I_{{\mathop{\rm int}} }}\biggl( {a_n}k\sin \theta \cos \phi w(\theta ,\phi ){F^0}\left( {\theta ,\phi } \right) \nonumber\\
	\times \sin \left( {kd_n^0\sin \theta \cos \phi  + {\beta _n}} \right) \biggr). \label{q}
	\end{align}
	In (\ref{e_opt}), $\boldsymbol{\mu_\mathcal{L}}$ is a vector of  Lagrangian multipliers, whose element $n$ is
	$\boldsymbol{\mu_\mathcal{L}}(n)=\mu_{n+1}-\mu_n$,
	with $\mu_n$ being a Lagrangian multiplier associated with constraint $n$. 
\end{theorem}

\begin{proof}
	See Appendix $A$.
\end{proof}

Using Theorem \ref{Theor1}, we can update the distance of each drone from the origin as follows:\vspace{-0.24cm}
\begin{equation}
{\boldsymbol{d}^1} = \boldsymbol{d}^0 + {\boldsymbol{e}^*}, \label{d_opt}
\end{equation}
where $\boldsymbol{d}^1=[d^1_n]_{N\times1}$, and $\boldsymbol{d}^0=[d^0_n]_{N\times1}$, $n\in \mathcal{N}$.

Clearly, $\boldsymbol{d}^1$ leads to a better solution than $\boldsymbol{d}^0=[d_n]_{N\times 1}$. In fact, we can proceed and further improve the solution to (\ref{d3}) by updating $\boldsymbol{d}^1$. In particular, at step update $r\in \mathds{N}$, we find  $\boldsymbol{d}^{(r)}$:\vspace{-0.20cm}
\begin{equation}
{\boldsymbol{d}^{(r)}} = \boldsymbol{d}^{(r-1)} + {\boldsymbol{e}^{*{(r)}}}, \vspace{-0.20cm}
\end{equation}
where $\boldsymbol{e}^{*{(r)}}$ is the optimal perturbation vector at step $r$ which is obtained based on $\boldsymbol{d}^{(r-1)}$. 

Note that, 
at each step, the objective function in (\ref{d3}) decreases. Since the objective function
is monotonically decreasing and bounded from below, the solution converges after several updates. We note that due to the approximation used in (\ref{approx}), the solution may not be a global optimal. Nevertheless, as we can see from Theorem \ref{Theor1}, it is analytically tractable and, hence, it has a low computational complexity.  Here, we use $\boldsymbol{d}^*$ to represent the vector of nearly-optimal distances of drones from the original of the array. Next, we use $\boldsymbol{d}^*$ to determine the optimal 3D locations of the drones that result in a maximum array directivity towards a given ground user.\vspace{-0.3cm}

\subsection{Optimal Locations of Drones}
Here, following from Subsection \ref{spaceP}, we derive the optimal 3D positions of drones that yields a maximum directivity of the drone-based antenna array.  Let $(x^\textrm{u} _{i},y^\textrm{u} _{i},z^\textrm{u}_i)$ and $(x_o,y_o,z_o)$ be, respectively, the 3D locations of user $i\in \mathcal{L}$ and the origin of the antenna array. 

Without loss of generality, we translate the origin of our coordinate system to the origin of the antenna array. In other words, we assume that the arrays' center is the origin of our translated coordinate system. In this case, the 3D location of user $i$ will be $(x^\textrm{u} _{i}-x_o,y^\textrm{u}_{i}-y_o,z^\textrm{u}_i-z_o)$. Subsequently, the polar and azimuthal angles of user $i$ in the spherical coordinate (with an origin of antenna array) are given by:\vspace{-0.15cm}
\begin{align} 
&{\theta _i} = {\cos ^{ - 1}}\left[ {\frac{{z_i^u - {z_o}}}{{\sqrt {{{(x_i^u - {x_o})}^2} + {{(y_i^u - {y_o})}^2} + {{(z_i^u - {z_o})}^2}} }}} \right], \label{tetai}\\
&{\phi _i} = {\sin ^{ - 1}}\left[ {\frac{{y_i^u - {y_o}}}{{\sqrt {{{(x_i^u - {x_o})}^2} + {{(y_i^u - {y_o})}^2}} }}} \right].
\end{align}

 Now, the optimal locations of the drones in the antenna array is given as follows.

\begin{theorem} \label{Theorem2}
	\normalfont
	The optimal locations of the drones for maximizing the directivity of the drone-based antenna array towards a given ground user will be:
		\begin{small}
		\begin{align} 
		&	\begin{pmatrix}x^*_m, y^*_m, z^*_m \end{pmatrix}^T= \nonumber\\
		& \hspace{-0.15cm}	\begin{cases}
		\boldsymbol{R}_\textrm{rot}\, \begin{pmatrix} d^*_m\sin{\alpha_o}\cos{\gamma_o}, d^*_m\sin{\alpha_o}\sin{\beta_o}, d^*_m\cos{\alpha_o}\end{pmatrix}^T,& \hspace{-0.4cm} \, m \le M/2, \\
		-\boldsymbol{R}_\textrm{rot}\, \begin{pmatrix} d^*_m\sin{\alpha_o}\cos{\gamma_o}, d^*_m\sin{\alpha_o}\sin{\gamma_o}, d^*_m\cos{\alpha_o}\end{pmatrix}^T\hspace{-0.1cm},&\hspace{-0.3cm} \, m> M/2,
		\end{cases} 
		\end{align}
	\end{small}

\noindent   	where $\alpha_o$ and $\gamma_o$ are the initial polar and azimuthal angles of drone $m\le M/2$ with respect to the array's center. $\boldsymbol{R}_\textrm{rot}$ is the rotation matrix for updating drones' positions, given by:
\begin{small}
	\begin{align}	
	&\boldsymbol{R}_\textrm{rot}= \nonumber \\  &\begin{pmatrix}
	{a_x^2(1 - \delta ) + \delta }&{{a_x}{a_y}(1 - \delta ) - \lambda {a_z}}&{{a_x}{a_z}(1 - \delta ) + \lambda {a_y}}\\
	{{a_x}{a_y}(1 - \delta ) + \lambda {a_z}}&{a_y^2(1 - \delta ) + \delta }&{{a_y}{a_z}(1 - \delta ) - \lambda {a_x}}\\
	{{a_x}{a_z}(1 - \delta ) - \lambda {a_y}}&{{a_y}{a_z}(1 - \delta ) + \lambda {a_x}}&{a_z^2(1 - \delta ) + \delta }
	\end{pmatrix},
	\end{align} 
\end{small} 
	where $\delta  = \left\| \boldsymbol{{{q_i} \cdot {q_{\max }}}} \right\|$, $\lambda  = \sqrt {1 - {\delta ^2}}$, $\boldsymbol{q_i}=\begin{pmatrix}
	{\sin {\theta_i} \cos {\phi_i}}\\
	{\sin {\theta_i} \sin {\phi_i}}\\
	{\cos {\theta_i}}
	\end{pmatrix}$, $\boldsymbol{q_{\textrm{max}}}= \begin{pmatrix} \sin {\theta_\textrm{max}} \cos {\phi_\textrm{max}}\\ \sin {\theta_\textrm{max}} \sin {\phi_\textrm{max}}\\ \cos {\theta_\textrm{max}} \end{pmatrix}$. Moreover,  $a_x$, $a_y$, and $a_z$ are the elements of vector $\boldsymbol{a}={\begin{pmatrix}
		{{a_x}},
		{{a_y}},
		{{a_z}}
		\end{pmatrix}^T}= \boldsymbol{{{q_i} \times {q_{\max }}}}$.

\end{theorem}

\begin{proof}
	See Appendix $B$.
\end{proof}

	\begin{figure}[!t]
		\begin{center}
			\vspace{-0.1cm}
			\includegraphics[width=8.0 cm]{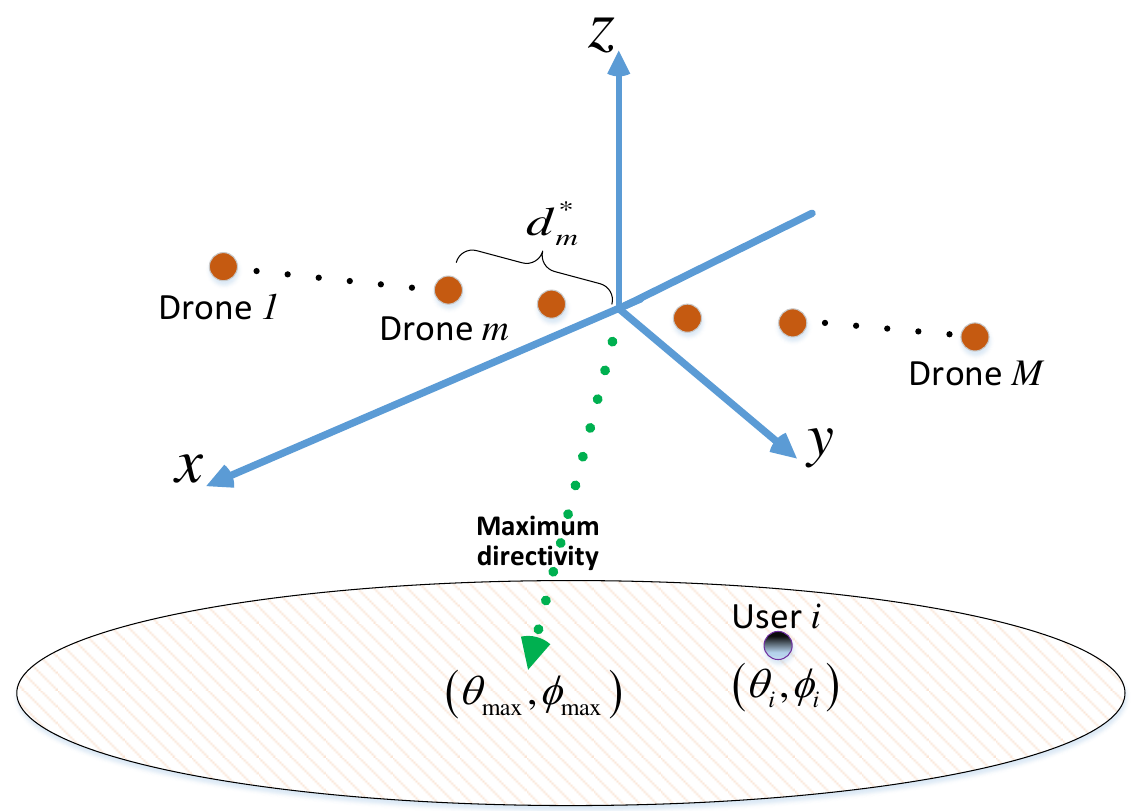}
			\vspace{-0.1cm}
			\caption{  \small Illustrative figure for Theorem 2. }\vspace{-.2cm}
			\label{Ther1Fig}
		\end{center}	
	\end{figure}

Using Theorem \ref{Theorem2}, we can find the optimal locations of the drones such that the directivity of the drone-based antenna array is maximized towards any given ground user. Moreover, this theorem can be used to dynamically update the drones' positions for beam steering while serving different ground users.  
 
 
Thus far, we have determined the optimal locations of the drones in the antenna array to maximize the directivity of the array towards any given ground user. Therefore, the data rate is maximized and, hence, the transmission time for serving the user is minimized. In Algorithm 1, we have summarized the key steps needed for optimizing the locations of drones with respect to the center of the array.

\setlength\textfloatsep{0.6\baselineskip plus 3pt minus 2pt}
\begin{algorithm} [t]
	\begin{small}
		\caption{Optimizing drones' locations for maximum array gain towards user $i$. }\label{Gradient}
		\begin{algorithmic}[1] 
			\State {\textbf{Inputs:} Locations of user $i$, $(x^\textrm{u} _{i},y^\textrm{u} _{i},z^\textrm{u}_i)$, and origin of array, $(x_o,y_o,z_o)$.}
			\State \textbf{Outputs:} Optimal drones' positions, $(x^*_{m,i},y^*_{m,i},z^*_{m,i})$, $\forall m \in \mathcal{M}$. 
			\State{Set initial values for distance between drones, $\boldsymbol{d}$.}
			\State {Find ${\boldsymbol{e}}^*$ by using  (\ref{e_opt})-(\ref{q}).} \label{stepe}
			\State {Update ${\boldsymbol{d}}$ based on (\ref{d_opt}).}\label{stepd}
			\State{Repeat steps (\ref{stepe}) and (\ref{stepd}) to find the optimal spacing vector $\boldsymbol{d}^*$.}
			\State { Use (\ref{tetai})-(\ref{R_rotation}) to determine $(x^*_m,y^*_m,z^*_m)$, $\forall m \in \mathcal{M}$. }

		\end{algorithmic} 
	\end{small} 
\end{algorithm}

Hence, using Algorithm 1, we can determine the optimal locations of the array's drones with respect to each ground user. To serve multiple users spread over a given geographical area, the drones must dynamically move between these determined optimal locations. This, in turn, yields a control time for drone movement that must be optimized. From (\ref{OPT1-1}), we can see that the service time decreases by reducing the control time. Therefore, next, using the determined drones' locations in Section III, we minimize the control time of the drones. 

\section{Time-Optimal Control of Drones } \label{TimeControlSec}
Here, our goal is to minimize the control time that the drones spend to move between the optimal locations which are determined in Section \ref{SecDirectivity}. \textcolor{black}{While moving the drone-based antenna array, we assume that the array rotates around its center in order to steer the beam and serve different users. Hence, the order of the drones (i.e., drones' indices) on the array does not change while moving the array. This approach significantly facilitates collision avoidance between the  drones as their paths do not intersect.}

In this section, we derive the optimal rotors' speeds for which the quadrotor drones can move and stabilize their positions within a minimum time. Moreover, we account for \emph{wind effects} while analyzing the drones' stability in the proposed drone-based antenna array system. \vspace{-0.3cm}


\subsection{Dynamic Model of a Quadrotor Drone}  
 
 Fig.\,\ref{Drones12} shows an illustrative example of a quadrotor drone. This drone has four rotors that can control the hovering and mobility of the drone. In particular, by adjusting the speed of these rotors, the drone can hover and move horizontally or vertically. Let $(x,y,z)$ be the 3D position of the drone. Also, we use $(\psi_\textrm{r},\psi_\textrm{p},\psi_\textrm{y})$ to represent the roll, pitch, and yaw angles that capture the orientation (i.e., attitude) of the drone. \textcolor{black}{Roll, pitch, and yaw are rotation angles defined with respect to the body frame. Here, the origin of the body frame coordinate system (represented by the $x_\textrm{b}$-$y_\textrm{b}$-$z_\textrm{b}$ axes)  is at the center of the drone, $x_\textrm{b}$ is along the arm between rotors 1 and 3,  $y_\textrm{b}$ is along the arm between rotors 2 and 4, and $z_\textrm{b}$ is in the direction of the cross product of the $x_\textrm{b}$ and $y_\textrm{b}$ axes. In this case, roll, pitch, and yaw, are rotations along  $x_\textrm{b}$,  $y_\textrm{b}$, and  $z_\textrm{b}$.}

  The speed of rotor $i\in\{1,2,3,4\}$ is given by $v_i$. For a quadrotor drone, the total thrust and torques that lead to the roll, pitch, and yaw movements are related to the rotors' speeds by \cite{BOOKControl}:\vspace{-0.15cm}
\begin{equation} \label{Velocit}
\left( {\begin{array}{*{20}{c}}
	{{T_\textrm{tot}}}\\
	{{\kappa _1}}\\
	{{\kappa _2}}\\
	{{\kappa _3}}
	\end{array}} \right) = \left( {\begin{array}{*{20}{c}}
	{{\rho _1}}&{{\rho _1}}&{{\rho _1}}&{{\rho _1}}\\
	0&{-l{\rho _1}}&0&{  l{\rho _1}}\\
	{-l{\rho _1}}&0&{  l{\rho _1}}&0\\
	{{-\rho _2}}&{  {\rho _2}}&{{-\rho _2}}&{  {\rho _2}}
	\end{array}} \right)\left( {\begin{array}{*{20}{c}}
	{v_1^2}\\
	{v_2^2}\\
	{v_3^2}\\
	{v_4^2}
	\end{array}} \right), 
 \end{equation}
 where $T_\textrm{tot}$ is the total thrust generated by the rotors. \textcolor{black}{The direction of the thrust is upward perpendicular to the rotors' plane, as we can see from Fig. \ref{Drones12}.} $\kappa_1$, $\kappa_2$, and $\kappa_3$ are the torques for roll, pitch and yaw movements. $\rho _1$ and $\rho _2$ are lift and torque coefficients, and $l$ is the distance from each rotor to the center of the drone.

\begin{figure}[t!]
	\centering
\hspace{-0.5cm}	\begin{subfigure}[t]{0.3\textwidth}
		\vspace{-0.01cm}
	\hspace{-0.5cm}		\includegraphics[width=4.6 cm]{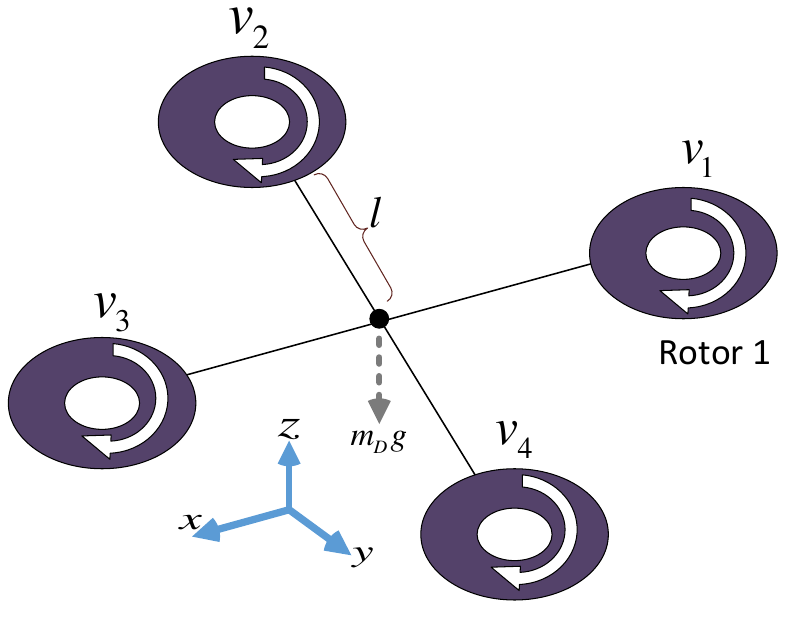}
		\vspace{-0.3cm}
		\label{Drone}
	\end{subfigure}%
	~ ~ ~ \hspace{-1.4cm}
	\begin{subfigure}[t]{0.3\textwidth}
		\vspace{-0.01cm}
\hspace{-0.5cm}		\includegraphics[width=4.6 cm]{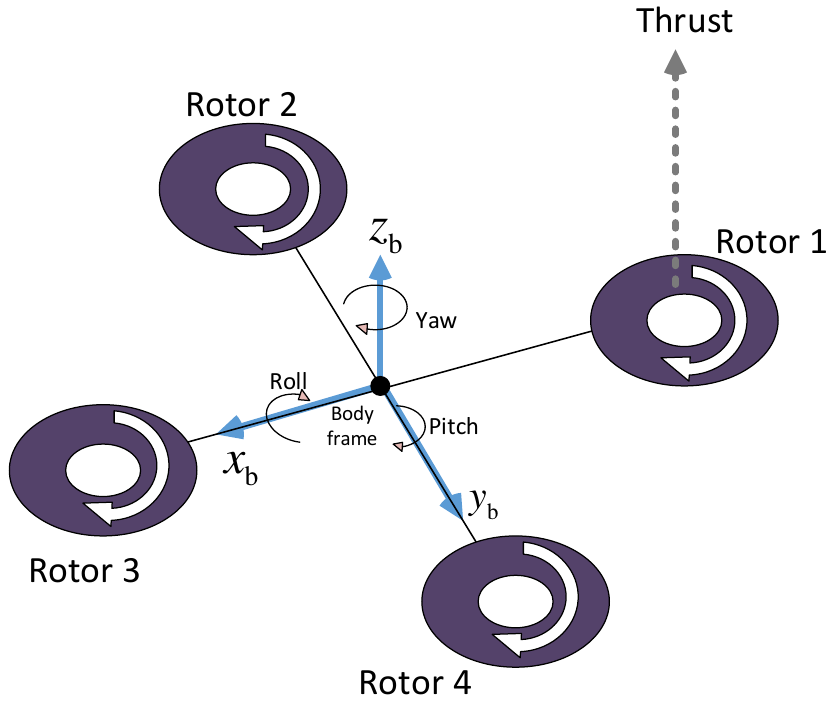}
	\vspace{-0.1cm}
	\end{subfigure}
	\caption{\textcolor{black}{\small A quadrotor drone.}} 	\label{Drones12}
\end{figure}

%
%
%

Now, we write the dynamic equations of a quadrotor drone in presence of an external wind force as follows\footnote{Note that, here, drag coefficients are assumed to be negligible.}\vspace{-0.2cm}:
\begin{align}
&\ddot x = \left( {\cos {\psi _r}\sin {\psi _p}\cos {\psi _y} + \sin {\psi _r}\sin {\psi _y}} \right)\frac{{{T_\textrm{tot}}}}{m_D}+\frac{F^\textrm{W}_x}{m_D}, \label{ax}\\
&\ddot y = \left( {\cos {\psi _r}\sin {\psi _p}\sin {\psi _y} + \sin {\psi _r}\cos {\psi _y}} \right)\frac{{{T_\textrm{tot}}}}{m_D}+\frac{F^\textrm{W}_y}{m_D},\label{ay}\\
&\ddot z = \left( {\cos {\psi _r}\cos {\psi _p}} \right)\frac{{{T_\textrm{tot}}}}{m_D} - g+\frac{F^\textrm{W}_z}{m_D},\label{az}\\
& \ddot \psi_\textrm{r}=\frac{\kappa_2}{I_x}, \label{pitc}\\
& \ddot \psi_\textrm{p}=\frac{\kappa_1}{I_y},\label{rol}\\
& \ddot \psi_\textrm{y}=\frac{\kappa_3}{I_z}, 
\end{align}
where $m_{D}$ is the mass of the drone, and $g$ is the gravity acceleration. $F^\textrm{W}_x$, $F^\textrm{W}_y$, and $F^\textrm{W}_z$ are the wind forces in positive $x$, $y$, and $z$ directions. Also, $I_x$, $I_y$, $I_z$ are constant values which represent the moments of inertia along $x$, $y$, and $z$ directions.
\textcolor{black}{From (25), we can see that the total thrust, $T_\textrm{tot}$ is directly related to the rotor speed.  Also, (26)-(28) capture the relationship between $T_\textrm{tot}$ and the drone's acceleration. Hence,
	using (25)-(28), we can find the drone's accelerations in the $x$, $y$, and $z$ directions. These accelerations  are directly related to position and velocity of the drone using classical kinematic equations \cite{hurtado2012kinematic}.}

Given the dynamic model of the drone, we aim to find the optimal speeds of the rotors such that the drone moves from an initial location $(x_I,y_I,z_I)$ to a new location $(x_D,y_D,z_D)$ within a minimum time duration. Under such optimal control inputs (i.e., rotors' speed), the time needed for each UAV to update its location based on the users' locations will be minimized. Note that the drone must be stationary at its new location and it does not move in $x$, $y$, or $z$ direction.
Let $(x(t), y(t),z(t))$ and $\left( {{\psi _\textrm{r}}(t),{\psi _\textrm{p}}(t),{\psi _\textrm{y}}(t)} \right)$ be the 3D location and orientation of the drone at time $t \in \left[ {0,{T_{I,D}}} \right]$, with $T_{I,D}$ being the total control time for moving from location $I$ to location $D$. Now, we can formulate our time-optimal control problem for a drone, moving from location $I$ to location $D$, as follows:\vspace{-0.1cm} 
\begin{align}
 &\mathop {{\mathop{\rm minimize}\nolimits} }\limits_{[{v_1}(t),{v_2}(t),{v_3}(t),{v_4}(t)]} {T_{I,D}}, \label{OptCont}\\
 \textrm{st.}\,\,\,& \textcolor{black}{|v_w(t)|\le v_\textrm{max},\,\,\, \forall w\in\{1,...,4\},} \label{Vmax}\\
&\left( {x(0),y(0),z(0)} \right) = \left( {{x_I},{y_I},{z_I}} \right),\label{Loc1}\\
 &\left( {x({T_{I,D}}),y({T_{I,D}}),z({T_{I,D}})} \right) = \left( {{x_D},{y_D},{z_D}} \right), \label{Loc2}\\
& \left( {\dot x(T_{I,D}),\dot y(T_{I,D}),\dot z(T_{I,D})} \right) = \left( {0,0,0} \right), \label{Veloc}
 \end{align}          
where $[{v_1}(t),{v_2}(t),{v_3}(t),{v_4}(t)]$ represents the rotors' speeds at time $t$. In (\ref{Vmax}), $v_\textrm{max}$ is the maximum possible speed of each rotor. Constraints (\ref{Loc1}) and (\ref{Loc2}) show the initial and final location of the drone (which are determined based on Algorithm 1), (\ref{Veloc}) indicates that the drone will be stationary at its final location. Here,  we assume  $\left( {{\psi _r}(0),{\psi _p}(0),{\psi _y}(0)} \right) = \left( {0,0,0} \right)$.

\textcolor{black}{In  (\ref{OptCont}), the goal is to minimize the control time that a drone needs in order to move between two locations, along a linear path. The objective function is the control time, and the optimization variables are the speeds of rotors. In (\ref{OPT1-1}), $T_{I,D}$ is the control time that a quadrotor drone spends to move from location $I$ to location $D$, the optimization variables are the speeds of rotors at time $t$, which are denoted by $v_1(t)$,  $v_2(t)$,  $v_3(t)$, and  $v_4(t)$.}   Note that in (\ref{OPT1-1}), the control time for serving user $i$, $T_i^\textrm{crl}$, is equal to the maximum control time among the drones that update their positions according to the user. 
 
Our problem in (\ref{OptCont}) is difficult to solve due to its non-linear nature, and coupled relation of the dynamic system parameters as well as the infinite number of optimization variables given the continuous time interval $[0,T_{I,D}]$. Consequently, in general, the exact analytical solution to such nonlinear time-optimal control problem may not be explicitly derived as pointed out in \cite{computationally} and \cite{TimeOptimalDrone}.  
 To provide a tractable solution to our time-optimal control problem in (\ref{OptCont}), we decompose the movements and orientation changes of drones. In particular, we minimize the time durations needed for orientation adjustment and displacement of the drone, separately. While this approach yields a suboptimal solution, it can be used to derive a closed-form expression for the control inputs (i.e., rotors' speeds) in (\ref{OptCont}) and, thus,  it is remarkably easy to implement. In addition, the computational time, which is a key constraint in wireless drone systems, can be6t reduced.

 Now, we aim to derive the optimal speeds of rotors for which the drone can update its locations within a minimum time duration. To this end, we first present the following lemma from control theory \cite{IntroControl} which will be then used to derive the optimal rotors' speeds. 
 
 \begin{lemma}\textbf{\normalfont (From \cite{IntroControl}):} \label{Bangbang}
 	\normalfont
 	Consider the state space equations for an object within time duration $[0,T]$:  \vspace{-0.1cm}
 	\begin{align}
 	&\dot {\boldsymbol{x}}(t) =\boldsymbol{ A}\boldsymbol{x}(t) + \boldsymbol{b}u(t),\,\, u_\textrm{min}\le u(t)\le u_\textrm{max}, \label{L1}\\
 	&	\boldsymbol{x}(0)=\boldsymbol{x}_1,\label{L2}\\
 	&	\boldsymbol{x}(T)=\boldsymbol{x}_2,\label{L3} 
 	\end{align}
 \noindent where $\boldsymbol{x}(t)\in\mathds{R}^{N_s}$ is the state vector of the object at time $t\in [0,T]$, $N_s$ is the number of state's elements.  $u(t)$ is a bounded control input with $u_\textrm{max}$ and $u_\textrm{min}$ being its maximum and minimum values. $\boldsymbol{A}\in\mathds{R}^{N_s\times N_s}$ and $\boldsymbol{b}\in\mathds{R}^{N_s}$  are given constant matrices.  $\boldsymbol{x}_1$ and $\boldsymbol{x}_2$ are the initial and final state of the object. Then, the optimal control input that leads to a minimum state update time $T^*$ is given by \cite{IntroControl}: \vspace{-0.1cm}
 	\begin{equation} \label{umax}
 	{u^*}(t) = 
 	\begin{cases}
 	{u_\textrm{max }},\,\,\,&t \le \tau,\\
 	{u_\textrm{min }},\,\,\, &t> \tau,
 	\end{cases}
 	\end{equation}
 	where $\tau$ is called the switching time at which the control input changes. In this case, the control time decreases by increasing $u_\textrm{max}$ and/or decreasing $u_\textrm{min}$.  
 	
 \end{lemma}
 Lemma \ref{Bangbang} provides the solution to the time-optimal control problem for a dynamic system which is characterized by (\ref{L1})-(\ref{L3}). In particular, the optimal control solution given in (\ref{umax}) is refereed to as \emph{bang-bang} solution \cite{IntroControl}. In this case, the optimal control input is always at its extreme value (i.e. maximum or minimum). Next, we provide a new lemma (Lemma 2) which will be used along with Lemma 1 to solve (\ref{OptCont}). \vspace{-0.1cm}

 \begin{lemma} \label{Lem1}
 		\normalfont
Consider a drone that needs to move towards a given location $D$ (as shown in Fig.\ref{Forces}), with a coordinate ${\boldsymbol{P}_D} = ({x_D},{y_D},{z_D})$, in presence of an external force ${\boldsymbol{F}_\textrm{ex}} = ({F_{\textrm{ex},x}},{F_{\textrm{ex},y}},{F_{\textrm{ex},z}})$. The drone's orientation that leads to a movement with the maximum acceleration towards $\boldsymbol{P}_D$ is:\vspace{-0.1cm} 
\begin{align}
&\psi _\textrm{p}^D = {\cos ^{ - 1}}\left[ {\frac{{A\cos {\theta _D} - |{\boldsymbol{F}_\textrm{ex}}|\cos {\theta_\textrm{ex}}}}{{F}}} \right],\\
&\psi _\textrm{r}^D = {\tan ^{ - 1}}\left( {\tan \beta \times \sin \psi _p^D} \right),\\
&\psi _\textrm{y}^D = 0,
\end{align}
where\\
$ \hspace{-0.4cm}A \hspace{-0.1cm} =  \hspace{-0.1cm}{\left[ \hspace{-0.05cm} {F^2 + |{\boldsymbol{F}_\textrm{ex}}{|^2} + 2F|{\boldsymbol{F}_\textrm{ex}}|\cos \left( {\gamma  + {{\sin }^{ - 1}}\hspace{-0.05cm}\left( {\frac{{|{\boldsymbol{F}_\textrm{ex}}|}}{{F}}\sin \gamma } \right)} \right)} \right]^{1/2}}\hspace{-0.05cm}$, $\beta={\phi _D} - {\sin ^{ - 1}}\left[ {\frac{{|{\boldsymbol{F}_\textrm{ex}}|\sin {\theta _\textrm{ex}}\sin \left( {{\phi _D} - {\phi _\textrm{ex}}} \right)}}{{F\sin \psi _p^D}}} \right]$, $\gamma  = {\cos ^{ - 1}}\left( {\frac{{{\boldsymbol{F}_\textrm{ex}}.{\boldsymbol{P}_D}}}{{|{\boldsymbol{F}_\textrm{ex}}||{\boldsymbol{P}_D}|}}} \right)$, and $F$ is the magnitude of the maximum force of the drone. $|{\boldsymbol{F}_\textrm{ex}}|$ represents the magnitude of vector ${\boldsymbol{F}_\textrm{ex}}$, ${\theta _\textrm{ex}} = {\cos ^{ - 1}}\left( {\frac{{{F_{\textrm{ex},z}}}}{{|{\boldsymbol{F}_\textrm{ex}}|}}} \right)$, ${\phi _\textrm{ex}} = {\tan ^{ - 1}}\left( {\frac{{{F_{\textrm{ex},y}}}}{{{F_{\textrm{ex},x}}}}} \right)$, ${\phi _D} = {\tan ^{ - 1}}\left( {\frac{{{y_D}}}{{{x_D}}}} \right)$, and ${\theta _D} = {\cos ^{ - 1}}\left( {\frac{{{z_D}}}{{|{\boldsymbol{P}_D}|}}} \right)$.  
 \end{lemma}
 \begin{proof}
 	See Appendix $C$.
 \end{proof}

 	\begin{figure}[!t]
 		\begin{center}
 			\vspace{-0.1cm}
 			\includegraphics[width=8.5 cm]{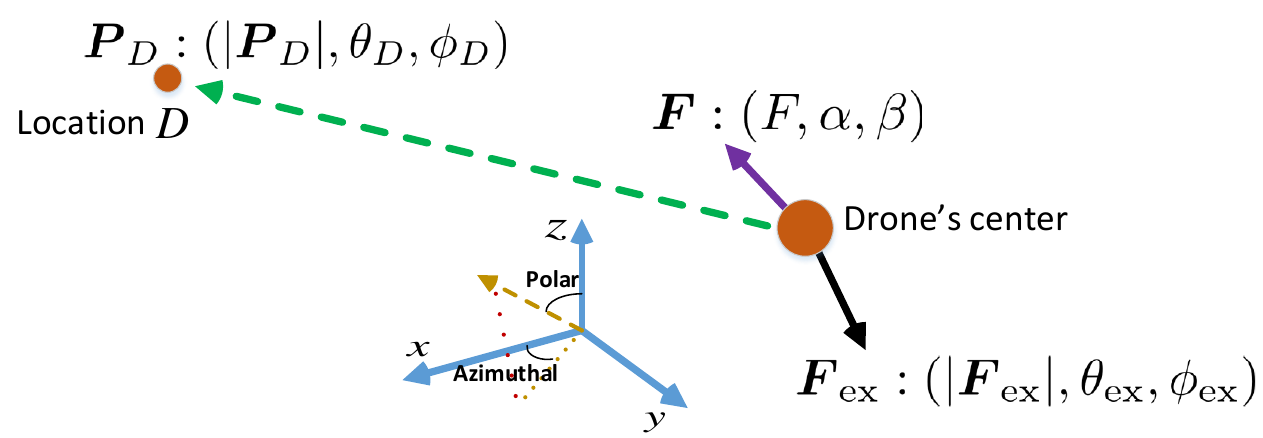}
 			\vspace{-0.1cm}
 			\caption{ \small Drone's movement in presence of an external force.}
 			\label{Forces}
 		\end{center}	\vspace{-0.1cm}
 	\end{figure} 
 	
Lemma \ref{Lem1}, can be used to determine the optimal orientation of the drone that enables it to move towards any given location in presence of external forces. Next, using Lemmas 1 and 2, we derive the speed of each drone's rotor for which the control time is minimized. In this case, we find the rotors' speeds at several pre-defined \emph{stages} in which the drone updates its position or orientation.\vspace{-0.1cm}  

\begin{theorem} \label{Thoerem3}
\normalfont
The optimal speeds of rotors with which a drone can move from location $(0,0,0)$, and $(0,0,0)$ orientation, to location $(x_D,y_D,z_D)$ within a minimum control time are given by: \vspace{-0.1cm}
\begin{equation} \label{Stage1}
\textrm{Stage 1: }\hspace{-0.18cm}
\begin{cases}
{v_2} = 0,{v_1} = {v_3} = \frac{1}{\sqrt{2}}{v_\textrm{max}},{v_4} = {v_{\max }}, &\hspace{-0.25cm}\textrm{if}\,\, 0 < t \le {\tau _1},\\
{v_4} = 0,{v_1} = {v_3} = \frac{1}{\sqrt{2}}{v_\textrm{max}},{v_2} = {v_{\max }} &\hspace{-0.25cm}\textrm{if}\,\, {\tau _1} < t \le {\tau _2},\\
{v_1} = 0,{v_2} = {v_4} = \frac{1}{\sqrt{2}}{v_\textrm{max}},{v_3} = {v_{\max }}, &\hspace{-0.25cm}\textrm{if}\,\, {\tau _2} < t \le {\tau _3},\\
{v_3} = 0,{v_2} = {v_4} = \frac{1}{\sqrt{2}}{v_\textrm{max}},{v_1} = {v_{\max }}, &\hspace{-0.25cm}\textrm{if}\,\, {\tau _3} < t \le {\tau _4}.
\end{cases}
\end{equation}
\begin{equation}
\textrm{Stage 2:\,\, } {v_1} = {v_2} = {v_3} = {v_4} = {v_{\max }},\,\,\ \textrm{if}\,\,\,{\tau _4} < t \le {\tau _5}.\vspace{-0.5cm}
\end{equation}

\begin{equation}
\textrm{Stage 3: }\hspace{-0.18cm}
\begin{cases}
{v_2} = 0,{v_1} = {v_3} = \frac{1}{\sqrt{2}}{v_\textrm{max}},{v_4} = {v_{\max }}, &\hspace{-0.25cm}\textrm{if}\,\, {\tau _5} < t \le {\tau _6},\\
{v_4} = 0,{v_1} = {v_3} = \frac{1}{\sqrt{2}}{v_\textrm{max}},{v_2} = {v_{\max }}, &\hspace{-0.25cm}\textrm{if}\,\,{\tau _6} < t \le {\tau _7},\\
{v_1} = 0,{v_2} = {v_4} = {v_{\max }},{v_3} = {v_{\max }}, &\hspace{-0.25cm}\textrm{if}\,\,{\tau _7} < t \le {\tau _8},\\
{v_3} = 0,{v_2} = {v_4} = \frac{1}{\sqrt{2}}{v_\textrm{max}},{v_1} = {v_{\max }}, &\hspace{-0.25cm}\textrm{if}\,\,{\tau _8} < t \le {\tau _9}.
\end{cases}\vspace{-0.2cm}
\end{equation}

\begin{equation}
\textrm{Stage 4:\,\, } {v_1} = {v_2} = {v_3} = {v_4} = {v_{\max }},\,\,\textrm{if}\,\,\,{\tau _9} < t \le {\tau _{10}}.\vspace{-0.5cm}
\end{equation}

\begin{equation}
\textrm{Stage\,5: }\hspace{-0.2cm}
\begin{cases}
{v_2} = 0,{v_1} = {v_3} = \frac{1}{\sqrt{2}}{v_\textrm{max}},{v_4} = {v_{\max }},&\hspace{-0.25cm}\textrm{if}\,\,{\tau _{10}} < t \le {\tau _{11}},\\
{v_4} = 0,{v_1} = {v_3} = \frac{1}{\sqrt{2}}{v_\textrm{max}},{v_2} = {v_{\max }},&\hspace{-0.25cm}\textrm{if}\,\,{\tau _{11}} < t \le {\tau _{12}},\\
{v_1} = 0,{v_2} = {v_4} = \frac{1}{\sqrt{2}}{v_\textrm{max}},{v_3} = {v_{\max }},&\hspace{-0.25cm}\textrm{if}\,\,{\tau _{12}} < t \le {\tau _{13}},\\
{v_3} = 0,{v_2} = {v_4} = \frac{1}{\sqrt{2}}{v_\textrm{max}},{v_1} = {v_{\max }},&\hspace{-0.25cm}\textrm{if}\,\,{\tau _{13}} < t \le {\tau _{14}}.
\end{cases}\vspace{-0.2cm}
\end{equation}

\begin{equation} \label{Stage6}
\textrm{Stage 6:\,\, } {v_1} = {v_2} = {v_3} = {v_4} = {v_\textrm{F}},\,\,\textrm{if}\,\,\, t>{\tau _{14}}. 
\end{equation}
Also, the total control time of the drone can be given by:
\begin{align} \label{TotalCont}
	T_{I,D}=&\sqrt{2d_{D}\Big(\frac{m_D}{A_{s2}}-\frac{m_D}{A_{s4}}\Big)}\nonumber\\
	&+\frac{2}{{{v_{\max }}}}\Big[\sqrt {\frac{{\Delta {\psi _{\textrm{p},1}}{I_y}}}{{l{\rho _1}}}}+\sqrt {\frac{{\Delta {\psi _{\textrm{r},1}}{I_x}}}{{l{\rho _1}}}}+\sqrt {\frac{{\Delta {\psi _{\textrm{p},3}}{I_y}}}{{l{\rho _1}}}}\nonumber\\
	&+\sqrt {\frac{{\Delta {\psi _{\textrm{r},3}}{I_x}}}{{l{\rho _1}}}}+\sqrt {\frac{{\Delta {\psi _{\textrm{p},5}}{I_y}}}{{l{\rho _1}}}}+\sqrt {\frac{{\Delta {\psi _{\textrm{r},5}}{I_x}}}{{l{\rho _1}}}}\Big],
\end{align}
where  $v_{\max }$, $v_\textrm{in}$, and $v_\textrm{F}$ are, respectively, the maximum, the initial, and the final speeds of rotors. $m_D$ is the drone's mass, $\Delta\psi _{\textrm{r},i}$ and $\Delta\psi _{\textrm{p},i}$ are the roll and pitch changes in Stage $i$. $d_D$ is the distance between the initial and final locations of the drone. $\tau_1,...,\tau_{14}$ are the switching times at which the rotors' speeds changes. The values of switching times and $v_F$ are provided in the proof of this theorem. 
\end{theorem}
\begin{proof}
	See Appendix $D$.
\end{proof}

In Theorem \ref{Thoerem3}, Stages 1, 3, and 5 correspond to the orientation changes, Stages 2 and 4 are related to the drone's displacement, and Stage 6 represents the drone's stability condition. Note that $v_\textrm{F}$ is adjusted such that the drone's stability is ensured at its final location.  In (\ref{TotalCont}), $A_{s2}$ and $A_{s4}$ are, respectively, the total forces towards the drone's destination at Stages 2 and 4. 

Using Theorem \ref{Thoerem3}, we can find the speeds of the rotors (at different time instances) that enable each to move towards its destination within a minimum time. The control time depends on the destination of the drone, external forces (e.g. wind and gravity), the rotors' speed, and the drone's weight. 

{\color {black}
\subsection{Collision Avoidance for Moving Drones}
First, we determine a situation in which collision between two drones when updating their locations is possible. Then, we propose a solution to avoid the collision situation.

Consider two adjacent drones that need to change their locations, as shown in Fig. \ref{Array1}. 
Clearly, the minimum distance between drones along their path is $x=d \sin \alpha$, where $\alpha$ and $d$ are shown in Fig. \ref{Array1}. In this case, if $x\ge D_\textrm{min}$, collision does not occur. Therefore, drones can move on a linear path without any collision However, if $x < D_\textrm{min}$, it is possible that the drones collide while they move. One way to avoid collision is to use non-straight paths for drones. For instance, an arc shape trajectory (as shown in Fig. \ref{Array2}) ensures that the distance between adjacent  drones remains above the minimum required distance, $D_\textrm{min}$.



\begin{figure}[t!]
	\centering
		\includegraphics[width=6 cm]{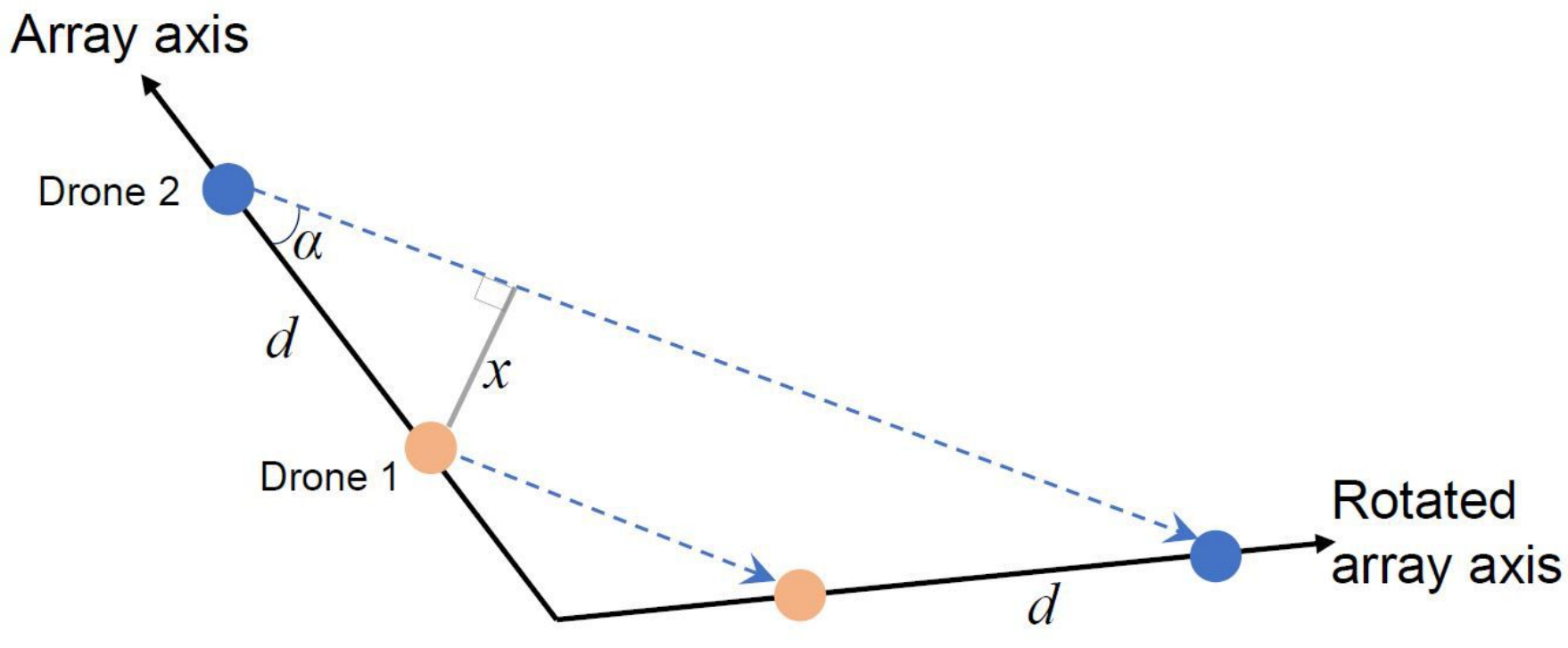}
		\caption{ \small \textcolor{black}{Drones' movements during the antenna array rotation (linear path).}}
		\label{Array1}
	\end{figure}%

	\begin{figure}[t]
		\centering
	\includegraphics[width=6 cm]{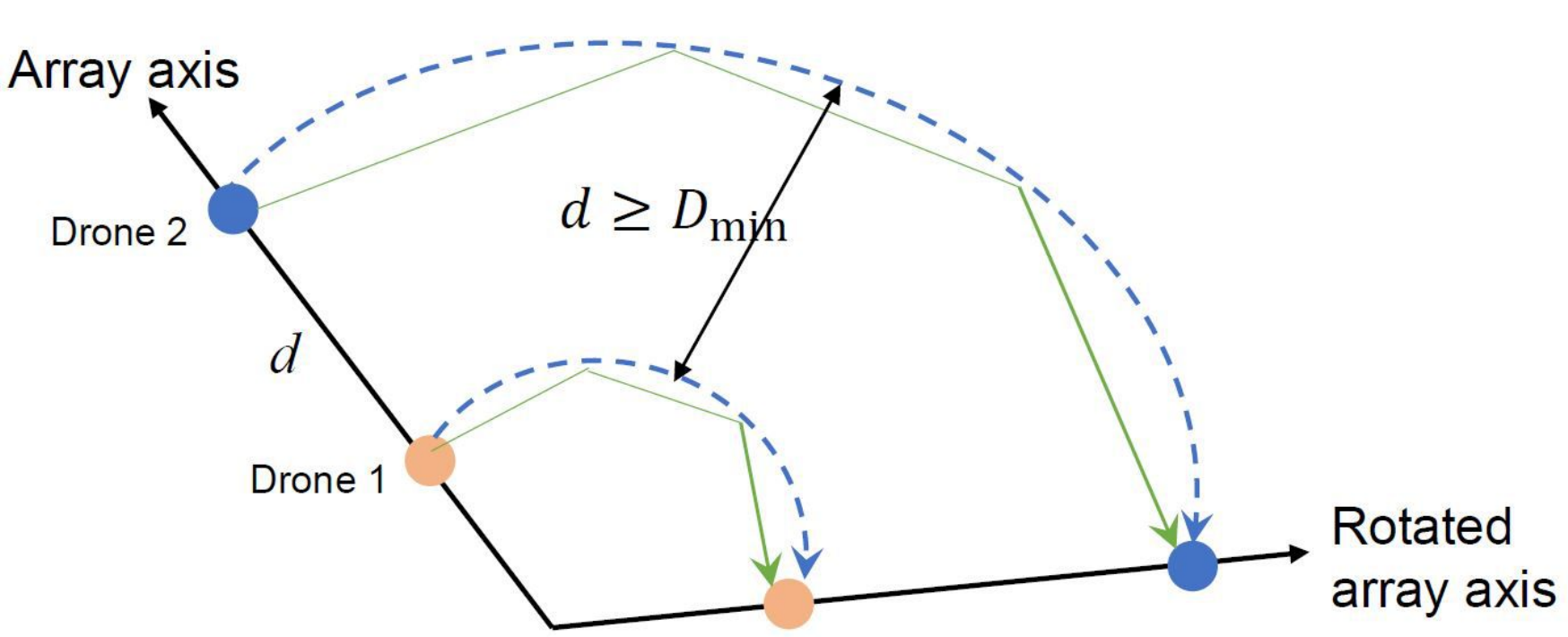}
		\caption{ \small \textcolor{black}{ Drones' movements during the antenna array rotation (arc path).}}
		\label{Array2}
\end{figure}

%

%
%
%

\subsection{User Scheduling Order}
Another factor that can impact the total control time of the drones is the user scheduling order. While any arbitrary user scheduling can be considered in our model, we adopt a scheduling order that yields a minimum total control time. To this end, we solve the following optimization problem which determines the optimal scheduling order:
\begin{align}
&\mathop {{\mathop{\rm minimize}\nolimits} }\limits_{[a_{ij}]_{L \times L}}\sum\limits_{i = 1,i \ne j}^L {\sum\limits_{j = 1}^L {{a_{ij}}{T_{ij}}} }, \label{obj} \\
\textrm{st.}\,\,\,& \sum\limits_{j = 1,j \ne i}^L {{a_{ij}} = 1} ,\,\,\, \forall i\in\mathcal{L},\,\, \sum\limits_{i = 1,i \ne j}^L {{a_{ij}} = 1},\,\,\, \forall j\in\mathcal{L},\label{C1}\\
&{a_{ij}} = \left\{ \begin{array}{l}
1 \,\,\, \textrm{if user $j$ is served after user $i$},\\
0 \,\,\,\textrm{otherwise} ,
\end{array} \right.
\end{align}   
where $L$ is the number of ground users in set $\mathcal{L}$, and $T_{ij}$ is the control time of drones when user $j$ is served after user $i$. $a_{ij}$ is a binary variable which is 1 if user $j$ is served after user $i$, and $[a_{ij}]_{L \times L}$ is a matrix that represents the scheduling order.  Constraint (\ref{C1}) indicates that each user is served only once. The optimization problem in (\ref{obj}) is a classical integer linear programming which can be solved using various methods such as a branch-and-bound algorithm \cite{lawler1966branch}.
 }


In summary, our approach for minimizing the service time, which is composed of the transmission time and the control time, is as follows. In the first step, using the approach in Section III, we minimize the transmission time for each ground user by optimizing the positions of drones with respect to the ground users. Then, based on these determined optimal drones' locations, we minimize the control time needed for adjusting the movement and orientations of drones. In Algorithm \ref{Algo2}, we summarize our approach for minimizing the service time.

\begin{algorithm} [t]
	\begin{small}
		\caption{Steps for minimizing the service time by solving (\ref{OPT1-1}). }\label{Algo2}
		\begin{algorithmic}[1] 
			\State {\textbf{Inputs:} Locations of users, $(x^\textrm{u} _{i},y^\textrm{u} _{i},z^\textrm{u}_i)$, $\forall i \in \mathcal{L}$, and origin of array, $(x_o,y_o,z_o)$.}
			\State \textbf{Outputs:} Optimal drones' positions, $(x^*_{m,i},y^*_{m,i},z^*_{m,i})$, rotors' speeds, \textcolor{black}{$v_{mw}(t)$, $\forall m \in \mathcal{M}, \forall i \in \mathcal{L}, w \in \{1,...,4\}$}, and total service time. 
			\State{Using Algorithm 1, find the optimal locations of drones with respect to each user, $(x^*_{m,i},y^*_{m,i},z^*_{m,i})$.}
			\State {Using Theorem\,3 and Lemma\,2, for each drone, determine the rotors' speeds for moving from $(x^*_{m,i-1},y^*_{m,i-1},z^*_{m,i-1})$ to $(x^*_{m,i},y^*_{m,i},z^*_{m,i})$.}
			\State {Compute the total service time based on (\ref{OPT1-1}), (\ref{OptCont}), and (\ref{TotalCont}).} 	
		\end{algorithmic} 
	\end{small} 
\end{algorithm}

\section{Simulation Results and Analysis} 

For our simulations, we consider a number of ground users uniformly distributed within a square area of size $1\,\text{km}\times 1\,\text{km}$. Unless stated otherwise, the number of users is 100, and the number of drones\footnote{In our simulations, each drone in the array has an omni-directional antenna, as in \cite{Garza, Weif}.} that form a linear array is assumed to be 10. The main simulation parameters are given in Table \ref{TableI}. \textcolor{black}{We compare the performance of our drone-based antenna array system with a case in which a drone-based antenna array uses a  fixed uniform drone separation, without any repositioning. For the benchmark, referred to as \emph{fixed-array} case, we consider half-wavelength drone spacing}\footnote{For the fixed-array case, we consider electronic beam steering with a 3 dB gain loss due to an imperfect phase synchronization.}. 

\begin{table}[!t]
	\normalsize
	\begin{center}
		\caption{ Main simulation parameters.}
		\vspace{-0.1cm}
		\label{TableI}
		\resizebox{9.0cm}{!}{
			\begin{tabular}{|c|c|c|}
				\hline
				\textbf{Parameter} & \textbf{Description} & \textbf{Value} \\ \hline \hline
				$f_c$	&     Carrier frequency     &      300\,MHz     \\ \hline 
				$P_i$	&     Drone transmit power     &     0.1\,\textrm{W}    \\ \hline
				
				$N_o$	&     Total noise power spectral density     &        -157\,dBm/Hz  \\ \hline
				$N$	&     Number of ground users     &   100 \\ \hline

				
				
				$(x_o,y_o,z_o)$	&  Array's center coordinate & (0,0,100) in meters \\ \hline
				
				$q_i$& Load per user & 100\,Mb \\ \hline
				
				$\alpha$	&     Pathloss exponent & 3 \\ \hline
				
				$I_x, I_y$	&    Moments of inertia & $4.9 \times 10^{-3} \textrm{kg.m}^2$  \cite{DroneParameter} \\ \hline
				
				$m_D$ &Mass of each drone & 0.5 kg \\ \hline
				
				$l$& Distance of a rotor to drone's center  & 20\,cm \\ \hline
				
				$\rho_1$ & lift  coefficient  & $2.9 \times 10^{-5}$ \cite{DroneParameter} \\ \hline
				
				$\beta_m-\beta_{m-1}$	&     Phase excitation difference for two adjacent antennas & $\frac{\pi}{5(M-1)}$\\ \hline
		\end{tabular}}
		
	\end{center}\vspace{-0.1cm}
\end{table}

{\color{black}First, we show an example on how the drones are separated in the proposed drone-based antenna array system. This result is provided in Table \ref{Table_Seperation} for two different carrier frequencies. 

%

\begin{table}[!t]
	\color{black}{	
		\begin{center}
			\caption{ \textcolor{black}{Separation distance of adjacent drones in an aerial antenna array with 10 drones.}}
			\vspace{-0.1cm}
			\label{Table_Seperation}
			\resizebox{9.0cm}{!}{
				\begin{tabular}{|c|c|c|}
					\hline
					Drones' separations (cm) & 	Drones' separations (cm), & Compared to wavelength\\, $f_c$=300 MHz, $\lambda$= 1 m & $f_c$=500 MHz, $\lambda$= 0.6\,m&  ($\lambda$)\\ \hline 			
					81.9 &	49.1 & 81.9 $\lambda$ \\ \hline 	
					88.7 &	53.2  & 88.7 $\lambda$  \\ \hline 				
				89.8 &	54.1  & 89.8 $\lambda$  \\ \hline 		
			90.7 &	 54.3  & 90.7 $\lambda$  \\ \hline 
		89.8&	 54.1  & 	89.8 $\lambda$  \\ \hline 		
		88.7&	 53.2  & 	88.7 $\lambda$  \\ \hline 	
		81.9&	 49.1   & 	81.9 $\lambda$  \\ \hline			
			\end{tabular}}
	\end{center}}\vspace{-0.1cm}
\end{table}

Fig. \ref{ServiceTime_BW} shows the total service time for the drone antenna array and the fixed-array case. For a given bandwidth, our proposed drone antenna array outperforms the fixed-array
case in terms of service time. This is due to the fact that, in the proposed approach, the drones' locations (and drone spacing)
are optimized such that the array antenna gain towards each user is maximized, hence reducing the
transmission time. Fig. \ref{ServiceTime_BW} also shows the tradeoff between bandwidth and service time. Clearly,
the service time decreases by using more bandwidth which effectively provides a higher data
rate.  Fig. \ref{ServiceTime_BW} shows that the drone antenna array improves spectral efficiency compared to the
fixed-array case. For instance, to achieve 10 minutes of service time, the drone antenna array
will require 32\% less bandwidth than in the fixed-array scenario.

\begin{figure}[!t]
	\begin{center}
		\vspace{-0.1cm}
		\includegraphics[width=8.0 cm]{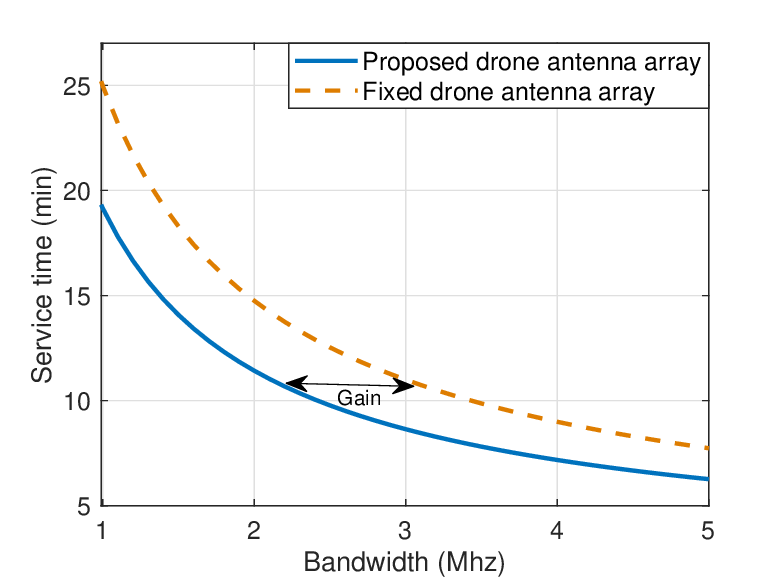}
		\vspace{-0.1cm}
		\caption{ \textcolor{black}{\small Service time vs. bandwidth for the drone antenna-array and fixed-array cases.} }\vspace{-0.4cm}
		\label{ServiceTime_BW}
	\end{center}	
\end{figure} 

In Fig. \ref{ServiceTime_numberofUsers}, we show the impact of the number of users on the service time. Clearly, the service
time increases as the number of users increases. For a higher number of users, the drones must deliver a higher data service which results in a higher transmission time. Moreover, in the proposed
drone antenna array case, the control time also increases while increasing the number of users.
Fig. \ref{ServiceTime_numberofUsers} shows that our proposed drone antenna array system outperforms the fixed-array case
for various number of users. For instance, using our approach, the average service time can be
reduced by 8 minutes (or 27\%) while serving 200 users.  Meanwhile, the users
can receive faster wireless services while exploiting the proposed drone antenna array system. 

\begin{figure}[!t]
	\begin{center}
		\vspace{-0.1cm}
		\includegraphics[width=7.6 cm]{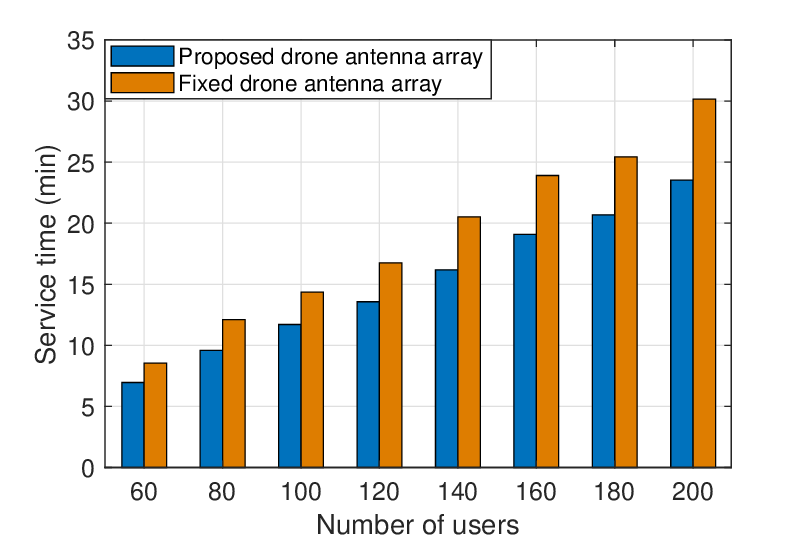}
		\vspace{-0.1cm}
		\caption{ \textcolor{black}{\small Service time vs. number of users for the drone antenna array and fixed-array (2MHz bandwidth).}}\vspace{-0.2cm}
		\label{ServiceTime_numberofUsers}
	\end{center}	
\end{figure}
}


Fig.\,\ref{Time_NumberOfDrones} shows how the control, transmission, and service times resulting from the proposed approach for different numbers of drones in the array. As the number of drones increases, the control time increases. In contrast, the transmission time (for 10 MHz bandwidth) decreases due to the increase of the array gain. Fig.\,\ref{Time_NumberOfDrones} shows that, by increasing the number of drones from 10 to 30, the average control time increases by 20\% while the average transmission time decreases by 36\%. Therefore, there is a tradeoff between the transmission time and the control time as a function of the number of drones in the array. 

\begin{figure}[!t]
	\begin{center}
		\vspace{-0.1cm}
		\includegraphics[width=8.0 cm]{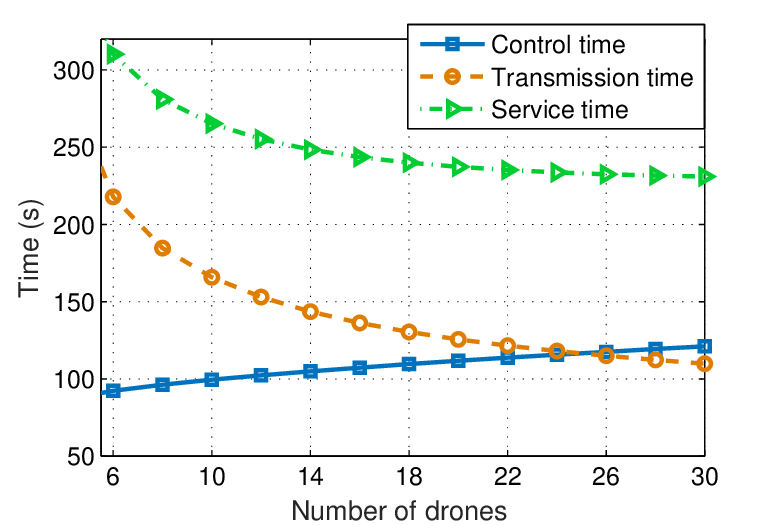}
		\vspace{-0.1cm}
		\caption{ \small Control, transmission, and service times vs. number of drones. }\vspace{-0.2cm}
		\label{Time_NumberOfDrones}
	\end{center}	
\end{figure}

In Fig.\,\ref{Control_users_Vmax=300and500}, we show how the number of users impacts the control time. As we can see from this figure, the control time increases while serving more users. This is due to the fact that, for a higher number of users, the drone-array must move more in order to steer its beam toward the users. The control time can be reduced by increasing the maximum speed of the rotors, which is in agreement with Theorem 3. For instance,  increasing the maximum rotors' speed from 300\,rad/s to 500\,rad/s yields around 35\% control time reduction when serving 200 users.

\begin{figure}[!t]
	\begin{center}
		\vspace{-0.1cm}
		\includegraphics[width=7.0 cm]{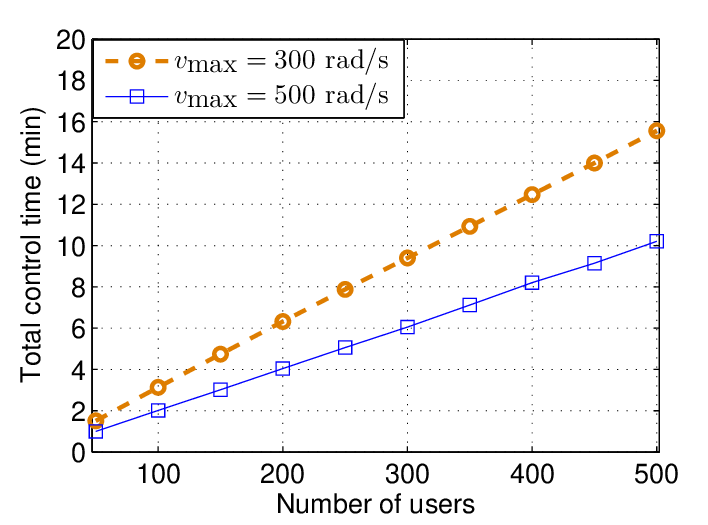}
		\vspace{-0.1cm}
		\caption{ \small Total control time vs. number of users. }\vspace{-0.2cm}
		\label{Control_users_Vmax=300and500}
	\end{center}	
\end{figure}

Fig.\,\ref{Wind_rotorSpeed} represents the speeds of the rotors needed to ensure the drone's stability in presence of wind, obtained using (\ref{vF_stable}). Clearly, the drone is stable when its total force which is composed of the wind force, gravity, and the drone force is zero. For ${F_\textrm{wind}} = |{F_\textrm{wind}}|\overrightarrow x$, the rotor's speed must increase as the wind force increases. In the ${F_\textrm{wind}} = |{F_\textrm{wind}}|\left( {\frac{1}{{\sqrt 3 }}\overrightarrow x  + \frac{1}{{\sqrt 3 }}\overrightarrow y  + \frac{1}{{\sqrt 3 }}\overrightarrow z } \right)$ case, however, the rotor's speed first decreases, and then increases. This is because, when $|{F_\textrm{wind}}|\le 3$\,N, the wind force helps hovering the drone by compensating for the gravity. Hence, the drone's force can be decreased by decreasing the speed of its rotors. For $|{F_\textrm{wind}}|> 3$\,N, the rotor's speed start increasing such that the total force on the drone becomes zero. This result also implies that, in some cases (depending on the magnitude and direction of wind), wind can facilitate hovering of the drone by overcoming the gravity force. However, in case of strong winds, the drone's stability may not be guaranteed by adjusting the speed of the rotors. This is because the drone force, which is limited by the maximum rotors' speeds, cannot overcome the external forces.\vspace{-0.1cm}
\begin{figure}[!t]
	\begin{center}
		\vspace{-0.5cm}
		\includegraphics[width=7.0 cm]{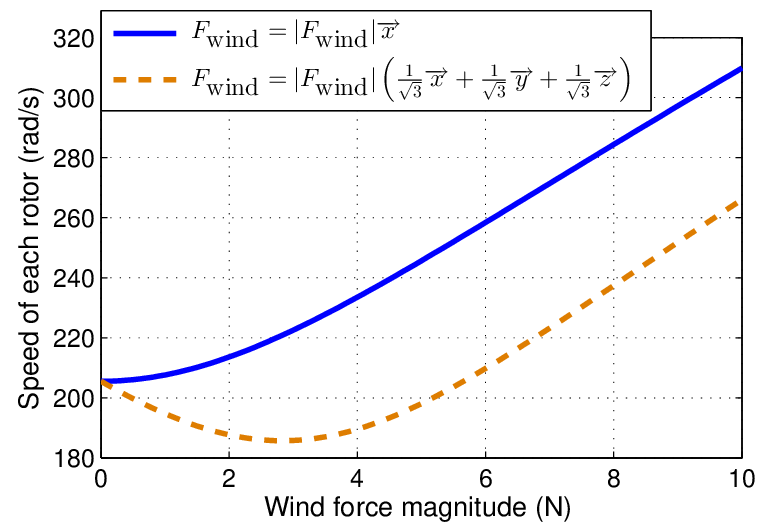}
		\vspace{-0.1cm}
		\caption{ \small Speed of each rotor vs. wind force under the drone's stability condition. }\vspace{-0.2cm}
		\label{Wind_rotorSpeed}
	\end{center}	
\end{figure}

 \section{Conclusion}\vspace{-0.10cm}
 In this paper, we have proposed a novel framework for employing a drone-enabled antenna array system that can provide wireless services to ground users within a minimum time. To this end, we have minimized the transmission time and the control time needed for changing the locations and orientations of the drones. First, we have optimized the positions of drones within the antenna array such that the transmission time for each user is minimized. Next, given the determined locations of drones,   we have minimized the control time of the quadrotor drones by optimally adjusting the rotors' speeds.   Our results have shown that the proposed drone antenna array with the optimal configuration yields a significant improvement in terms of the service time, spectral and energy efficiency. Our results have revealed key design guidelines and fundamental tradeoffs for leveraging in an antenna array system. To our best knowledge, this is the first comprehensive study on the joint communications and control of drone antenna array systems.

\section* {Appendix}

\subsection{Proof of Theorem 1}
 	First, we find $F^2(\theta,\phi)$ by using (\ref{approx}):
 	\begin{align}
 	& F^2(\theta,\phi)= 
 	{\left[ {2{F^0}(\theta ,\phi )} \right]^2}\hspace{-0.10cm} + \nonumber\\\hspace{-0.10cm} &{\left[ {2\sum\limits_{n = 1}^N {{a_n}k{e_n}\sin \theta \cos \phi \sin \left( {kd_n^0\sin \theta \cos \phi  + {\beta _n}} \right)} } \right]^2} \nonumber\\
 	- 8{F^0}&(\theta ,\phi )\sum\limits_{n = 1}^N {{a_n}k{e_n}\sin \theta \cos \phi \sin \left( {kd_n^0\sin \theta \cos \phi  + {\beta _n}} \right)}.\nonumber
 	\end{align}
 	Subsequently, our objective function in (\ref{d3}) can be written as:
 	\begin{align}
 	{I_{{\mathop{\rm int}} }}&\left( {{F^2}(\theta ,\phi ){w^2}(\theta ,\phi )} \right)= \nonumber\\
 	&4\left[ {{\boldsymbol{e}^T}\boldsymbol{G}\boldsymbol{e} - 2{\boldsymbol{e}^T}\boldsymbol{q} + {I_{{\mathop{\rm int}} }}\left( {F_0^2(\theta ,\phi ){w^2}(\theta ,\phi )} \right)} \right], \label{Iint}
 	\end{align}
 	where $\boldsymbol{G}$ and $\boldsymbol{q}$ are given in (\ref{G}) and (\ref{q}). 
 	Clearly, (\ref{Iint}) is a quadratic function of $\boldsymbol{e}$. Therefore, (\ref{Iint}) is convex if and only if $\boldsymbol{G}$ is a positive semi-definite matrix. 
 	Given (\ref{G}), we have:\vspace{-0.1cm}
 	\begin{equation}
 	{\boldsymbol{y}^T}\boldsymbol{G}y = \sum\limits_{n = 1}^N {{y_n}\sum\limits_{m = 1}^N {{y_m}{g_{m,n}}} }. \label{yGy}
 	\end{equation}
 	Now, in (\ref{G}), let us define \vspace{-0.3cm}
 	\begin{equation}
 	\hspace{-0.2cm}{z_n} = {a_n}k\sin \theta \cos \phi w(\theta ,\phi )\sin \left( {kd_n^0\sin \theta \cos \phi  + {\beta _n}} \right),
 	\end{equation}
 	then, using (\ref{yGy}), we have:\vspace{-0.2cm}
 	\begin{equation}
 	{\boldsymbol{y}^T}\boldsymbol{G}\boldsymbol{y} = {I_{{\mathop{\rm int}} }}\left( {{{\left[ {\sum\limits_{n = 1}^N {{z_n}{y_n}} } \right]}^2}} \right).\label{sum}
 	\end{equation}
 	In (\ref{IintF}), we can see that ${I_{{\mathop{\rm int}} }}(x) \ge 0$ for $x\ge0$. Hence, from (\ref{sum}), we can conclude that ${\boldsymbol{y}^T}\boldsymbol{G}\boldsymbol{y} \ge 0$. Therefore, $\boldsymbol{G}$ is positive semi-definite and the objective function in (\ref{d3}) is convex. Moreover, the constraints in (\ref{emin}) are affine functions which are convex. Hence, this optimization problem is convex.  Now, we find the optimal perturbation vector $\boldsymbol{e}$ by using Karush-Kuhn-Tucker (KKT) conditions. The Lagrangian function will be: 
 	\begin{align}
 	\mathcal{L}&={{\boldsymbol{e}^T}\boldsymbol{G}\boldsymbol{e} - 2{\boldsymbol{e}^T}\boldsymbol{q} 
 		+ {I_{{\mathop{\rm int}} }}\left( {F_0^2(\theta ,\phi ){w^2}(\theta ,\phi )} \right)} \nonumber \\
 	&+\sum\limits_{n = 1}^{N - 1} {{\mu _n}\left( {{e_n} - {e_{n + 1}} + {D_{\textrm{min} }} + d_n^0 - d_{n + 1}^0} \right)}, 
 	\end{align} 
 	where $\mu_n \ge 0$, $n=1,...,N-1$ are the Lagrange multipliers. 
 	
 	The necessary and sufficient (due to the convexity of the problem) KKT conditions for finding the optimal perturbation vector $\boldsymbol{e}$ are given by:\vspace{-0.4cm}
 	\begin{equation} \label{LAG}
 	\nabla_{\boldsymbol{e}}\left[\mathcal{L}\right]=0,\vspace{-0.2cm}
 	\end{equation}
 	which leads to $\boldsymbol{e} = {\boldsymbol{G}^{ - 1}}[\boldsymbol{q}+\boldsymbol{\mu_\mathcal{L}}]$, with $\boldsymbol{\mu_\mathcal{L}}$ being a $(N-1)\times 1$ vector whose element $n$ is
 	$\boldsymbol{\mu_\mathcal{L}}(n)=\mu_{n+1}-\mu_n$. Based on the complementary slackness conditions, we have:\vspace{-0.12cm}
 	\begin{equation} \label{SLAK}
 	\begin{cases}
 	&\hspace{-0.3cm} \mu_n\left( {{e_n} - {e_{n + 1}} + {D_{\textrm{min} }} + d_n^0 - d_{n + 1}^0} \right)=0, \,\, \forall n \in \mathcal{N}\,\backslash{\{N\}},\\
 	&\hspace{-0.3cm} \mu_n \ge 0, \,\, \forall n \in \mathcal{N}\,\backslash{\{N\}}.
 	\end{cases}
 	\end{equation}
 	Finally, the optimal perturbation vector, $\boldsymbol{e}^*$, can be determined by solving (\ref{LAG}) and (\ref{SLAK}).
 
 \subsection{Proof of Theorem 2}
 	In Subsection III-A, we have derived the optimal distance of drones from the origin that leads to a maximum array directivity. First, we consider an initial (or arbitrary) orientation,   as shown in Figure \ref{Ther1Fig}. Let $d^*_m$ be the optimal distance of drone $m\le M/2$ from the array's center, $\alpha_o$ and $\gamma_o$ be the initial polar and azimuthal angles of the drone. Based on the considered drones' locations, let $(\theta _{\max },\phi _{\max })=\textrm{argmax}\big[  F(\theta ,\phi )w(\theta ,\phi )\big]$ be a direction at which the directivity of the array is maximized. Our goal is to achieve the maximum directivity at a given direction $(\theta_i ,\phi_i)$ corresponding to user $i$. Therefore, we need to change the locations of the drones such that  $\theta_i=\theta_\textrm{max}$, and 	$\phi_i=\phi_\textrm{max}$. To this end, we align the unit vector  $(1,\theta_\textrm{max}, \phi_\textrm{max})$ with $(1,\theta_i, \phi_i)$  in the spherical coordinate and, then, we update the drones' positions accordingly. In the Cartesian coordinate system, we need to rotate vector $\boldsymbol{q_{\textrm{max}}}= \begin{pmatrix} \sin {\theta_\textrm{max}} \cos {\phi_\textrm{max}}, \sin {\theta_\textrm{max}} \sin {\phi_\textrm{max}}, \cos {\theta_\textrm{max}} \end{pmatrix}^T$ such that it becomes aligned with $\boldsymbol{q_i}=\begin{pmatrix}
 	{\sin {\theta_i} \cos {\phi_i}},
 	{\sin {\theta_i} \sin {\phi_i}},
 	{\cos {\theta_i}}
 	\end{pmatrix}^T$\!\!.

 	The rotation matrix for rotating a vector $\boldsymbol{u}$ about another vector $\boldsymbol{a}={\begin{pmatrix}
 		{{a_x}},
 		{{a_y}},
 		{{a_z}}
 		\end{pmatrix}^{\!T}}\!\!,$ with a $\omega$ rotation angle, is \cite{RotationOrentation}:
 	\begin{equation} \label{R_rotation}	
 	\!\!\! \!\!\! \boldsymbol{R}_\textrm{rot}\!= { \begin{pmatrix}
 		\boldsymbol{R}_{\textrm{rot},1}&
 		\boldsymbol{R}_{\textrm{rot},2}&
 		\boldsymbol{R}_{\textrm{rot},3}
 		\end{pmatrix}},
 	\end{equation} 
 	where \begin{small}$\boldsymbol{R}_{\textrm{rot},1}\!\!=\!\!\begin{pmatrix}
 		a_x^2(1 - \cos\omega ) + \cos\omega \\ {{a_x}{a_y}(1 - \cos\omega ) + {a_z} \sin\omega }\\ {{a_x}{a_z}(1 - \cos\omega) - {a_y} \sin\omega } \end{pmatrix}$,  $\boldsymbol{R}_{\textrm{rot},2}\!\!=\!\!\begin{pmatrix}
 		{{a_x}{a_y}(1 - \cos\omega ) - {a_z} \sin\omega } \\ {a_y^2(1 - \cos\omega ) + \cos\omega }\\ {{a_y}{a_z}(1 - \cos\omega ) + {a_x} \sin\omega} \end{pmatrix}$, and $\boldsymbol{R}_{\textrm{rot},3}\!\!=\!\!\begin{pmatrix}
 		{{a_x}{a_z}(1 - \cos\omega ) + {a_y} \sin\omega } \\ {{a_y}{a_z}(1 - \cos\omega ) - {a_x} \sin\omega }\\ {a_z^2(1 - \cos\omega ) + \cos\omega } \end{pmatrix}$\end{small}.


 In our problem, the rotation between $\boldsymbol{q_{\textrm{max}}}$ and $\boldsymbol{q_i}$ can be done about the normal vector of these vectors, with the rotation angle being the angle between $\boldsymbol{q_{\textrm{max}}}$ and $\boldsymbol{q_i}$. Hence, based on the dot-product and cross-product of vectors, we use $\boldsymbol{a}= \boldsymbol{{{q_i} \times {q_{\max }}}}$, and $\omega= \cos^{-1}(\boldsymbol{{{q_i} \cdot {q_{\max}}}})$ to find the rotation matrix in (\ref{R_rotation}). Now, we update the locations of drones using the rotation matrix. Clearly, for $m \le M/2$, the initial location of drone $m$ in the Cartesian coordinate is $\begin{pmatrix} d^*_m\sin{\alpha_o}\cos{\gamma_o}, d^*_m\sin{\alpha_o}\sin{\beta_o}, d^*_m\cos{\alpha_o}\end{pmatrix}^T$. As a result, the optimal locations of drones for serving user $i$ is given by:
 	\begin{align}
 	&{\begin{pmatrix}x^*_m, y^*_m, z^*_m \end{pmatrix}^T}= \nonumber\\  &\boldsymbol{R}_\textrm{rot} {\begin{pmatrix} d^*_m\sin{\alpha_o}\cos{\gamma_o}, d^*_m\sin{\alpha_o}\sin{\beta_o}, d^*_m\cos{\alpha_o}\end{pmatrix}^T}\hspace{-0.1cm},\nonumber \\
 	&\textrm{if}\,\, m \le M/2.
 	\end{align}
 	Finally, due to the symmetric configuration of the antenna array about the origin, the optimal locations of drones $m$ when $m>M/2$ are as follows:
 	\begin{align}
 	&{\begin{pmatrix}x^*_m, y^*_m, z^*_m \end{pmatrix}^T}= \nonumber \\ &-\boldsymbol{R}_\textrm{rot}\, {\begin{pmatrix} d^*_m\sin{\alpha_o}\cos{\gamma_o}, d^*_m\sin{\alpha_o}\sin{\beta_o}, d^*_m\cos{\alpha_o}\end{pmatrix}^T}, \nonumber \\
 	&\textrm{if}\,\, m \le M/2.
 	\end{align}
 	This completes the proof.
	
\subsection{Proof of Lemma 2}
 	To maximize the drone's acceleration towards the given location $D$, we need to maximize the total force in the direction of $\boldsymbol{P}_D$. Considering the center of the drone as the origin of the Cartesian and spherical coordinate systems, we can present the vectors of forces and the movement as in Fig.\,\ref{Forces}. In this figure, based on the Cartesian-to-spherical coordinates transformation, the polar and azimuthal angles in the spherical coordinate are given by ${\theta _\textrm{ex}} = {\cos ^{ - 1}}\left( {\frac{{{F_{\textrm{ex},z}}}}{{|{\boldsymbol{F}_\textrm{ex}}|}}} \right)$, ${\phi _\textrm{ex}} = {\tan ^{ - 1}}\left( {\frac{{{F_{\textrm{ex},y}}}}{{{F_{\textrm{ex},x}}}}} \right)$, ${\phi _D} = {\tan ^{ - 1}}\left( {\frac{{{y_D}}}{{{x_D}}}} \right)$, and ${\theta _D} = {\cos ^{ - 1}}\left( {\frac{{{z_D}}}{{|{\boldsymbol{P}_D}|}}} \right)$. Let $\alpha$ and $\beta$ be, respectively, the polar and azimuthal angles of the drone's force. Here, we seek to determine $\alpha$ and $\beta$ such that the drone can move towards location $D$ with a maximum acceleration (i.e., maximum total force). In this case, the total force $\boldsymbol{F}_\textrm{ex}+\boldsymbol{F}$ must be in the same direction as $\boldsymbol{P}_D$. Let $\gamma$ be the angle between $\boldsymbol{F}$ and $\boldsymbol{P}_D$, and $q$ be the angle between $\boldsymbol{F}_\textrm{ex}$ and $\boldsymbol{P}_D$. To ensure that $\boldsymbol{F}_\textrm{ex}+\boldsymbol{F}$ is in the direction of $\boldsymbol{P}_D$, we should have:\vspace{-0.35cm}
 	\begin{equation}
 	|{\boldsymbol{F}_{{\textrm{ex}}}}|\sin \gamma  = |\boldsymbol{F}|\sin q = F\sin q.
 	\end{equation}
 	Also, using the inner product formula, $\gamma$ is given by:
 	\begin{equation}
 	\gamma  = {\cos ^{ - 1}}\left( {\frac{{{\boldsymbol{F}_\textrm{ex}}\cdot{\boldsymbol{P}_D}}}{{|{\boldsymbol{F}_\textrm{ex}}||{\boldsymbol{P}_D}|}}} \right).
 	\end{equation}
 	As a result, $q$ will be:\vspace{-0.2cm}
 	\begin{equation}
 	{{q = }}{\sin ^{ - 1}}\left( {\frac{{|{\boldsymbol{F}_{ex}}|}}{{|\boldsymbol{F}|}}\sin \left[ {{{\cos }^{ - 1}}\left( {\frac{{{\boldsymbol{F}_{ex}}.{\boldsymbol{P}_D}}}{{|{F_{ex}}||{P_D}|}}} \right)} \right]} \right).
 	\end{equation}
 	Now, based on the law of cosines, the total force magnitude is equal to:
 	\begin{align} \label{A_Force}
 &	A \mathop  = \limits^\Delta |\boldsymbol{F}_\textrm{ex}+\boldsymbol{F}|= \nonumber\\
 	&{\left[ {F^2 + |{\boldsymbol{F}_\textrm{ex}}{|^2} + 2F|{\boldsymbol{F}_\textrm{ex}}|\cos \left( {\gamma  + {{\sin }^{ - 1}}\left( {\frac{{|{\boldsymbol{F}_\textrm{ex}}|}}{{F}}\sin \gamma } \right)} \right)} \right]^{1/2}}.
 	\end{align}
 	
 	By projection $(\boldsymbol{F}_\textrm{ex}+\boldsymbol{F})$, $\boldsymbol{F}_\textrm{ex}$, and $\boldsymbol{F}$ on $z$-axis and $x-y$ plane, we have:
 	\begin{align}
 	&A\cos {\theta _D} = |{\boldsymbol{F}_\textrm{ex}}|\cos {\theta _\textrm{ex}} + F\cos \alpha, \\
 	&	|{\boldsymbol{F}_\textrm{ex}}|\sin {\theta _\textrm{ex}}\sin \left( {{\phi _D} - {\phi _\textrm{ex}}} \right) = F\sin \alpha \sin \left( {{\phi _D} - \beta } \right).
 	\end{align}
 	Subsequently, we obtain $\alpha$ and $\beta$ as follows:
 	\begin{align}
 	&\alpha = {\cos ^{ - 1}}\left[ {\frac{{A\cos {\theta _D} - |{\boldsymbol{F}_\textrm{ex}}|\cos {\theta_\textrm{ex}}}}{{F}}} \right],\\
 	&\beta = {\phi _D} - {\sin ^{ - 1}}\left[ {\frac{{|{\boldsymbol{F}_\textrm{ex}}|\sin {\theta _\textrm{ex}}\sin \left( {{\phi _D} - {\phi _\textrm{ex}}} \right)}}{{F\sin \psi _p^D}}} \right].
 	\end{align}
 	Finally, considering the fact that the drone's force is perpendicular to its rotors' plane, as well as using the transformation between body-frame and earth-frame,  the drone's orientation can be given by\footnote{We consider $(0,0,0)$ as the initial  orientation. To change the orientation, we first update the pitch and, then, update the roll.}\vspace{-0.1cm}:
 	\begin{equation}
 	\psi _\textrm{p}^D=\alpha, 	\psi _\textrm{r}^D={\tan ^{ - 1}}\left( {\tan \beta \times \sin \psi _p^D} \right), 	\psi _\textrm{y}^D=0,
 	\end{equation}
 	which proves Lemma \ref{Lem1}.\vspace{-0.1cm}

 \subsection{Proof of Theorem 3}
 	Let $s(t)$ be the distance that the drone moves towards destination $D$ at time $t$. We define state $\boldsymbol{g}(t)=\left[
 	{s(t)} , {\dot s(t)} \right]^T$, and provide the following equation:
 	\begin{equation} \label{gt}
 	\boldsymbol{\dot g}(t) = \left[ {\begin{array}{*{20}{c}}
 		0&1\\
 		0&0
 		\end{array}} \right]\boldsymbol{g}(t) + \left[ {\begin{array}{*{20}{c}}
 		0\\
 		1
 		\end{array}} \right]{a_D}(t),
 	\end{equation}
 	where $a_\textrm{min}\le{a_D}(t)\le a_\textrm{max}$ is the drone's acceleration towards $D$, with $a_\textrm{min}$ and $a_\textrm{max}$ being the minimum and maximum values of ${a_D}(t)$. Clearly, the drone can reach the destination and stop at $D$ within duration $T$, if $\boldsymbol{g}(T)=[0,0]^T$. Based on Lemma \ref{Bangbang}, $T$ is minimized when ${a_D}(t) = 
 	\begin{cases}
 	{a_\textrm{max }},\,\,\,&0< t \le \tau,\\
 	{a_\textrm{min }},\,\,\,  &\tau<t\le T.
 	\end{cases}$.
 	Now, we find $\tau$ by using kinematic equations that describe an object's motion. Let $d_{D}$ be the distance between the initial and the final locations of the drone. Clearly, the drone's displacement until $t=\tau$ is equal to $\frac{1}{2}a_\textrm{max}\tau^2$. During $\tau<t\le T$, the displacement will be $\frac{1}{2}a_\textrm{min}(T-\tau)^2+a_\textrm{max}\tau(T-\tau)$. Hence, the total drone's disparagement is: 
 	\begin{equation} \label{kene1}
 	d_{D}=\frac{1}{2}a_\textrm{max}\tau^2+\frac{1}{2}a_\textrm{min}(T-\tau)^2+a_\textrm{max}\tau(T-\tau).
 	\end{equation}
 	
 	Also, considering the fact that drone stops (i.e. zero speed) at $t=T$, we have: \vspace{-0.2cm}
 	\begin{equation}\label{kene2}
 	a_\textrm{max}\tau+a_\textrm{min}(T-\tau)=0,
 	\end{equation}
 	According to (\ref{kene1}) and (\ref{kene2}), the total control time, $T$, and the switching time can be found by:
 	\begin{align}
 	&T=\sqrt{2d_{D}(\frac{1}{a_\textrm{max}}-\frac{1}{a_\textrm{min}})},\label{bigT}\\
 	&\tau= \frac{a_\textrm{min}}{a_\textrm{min}-a_\textrm{max}}T.
 	\end{align}
 	
 	As we can see from (\ref{bigT}), $T$ can be minimized by maximizing $a_\textrm{max }$ and minimizing $a_\textrm{min}$. To this end, we will adjust the drone's orientation as well as the rotors' speeds. Each drone's orientation can be determined by using Lemma \ref{Lem1}. Also, given (\ref{Velocit})-(\ref{az}), we can show that the optimal speeds of the rotors are ${v_1} = {v_2} = {v_3} = {v_4} = {v_{\max }}$.  
 	
 	To adjust the drone's orientation within a minimum time, we minimize the time needed for the pitch and roll updates. Using a similar approach as in (\ref{gt}), and considering (\ref{Velocit}), (\ref{pitc}), (\ref{rol}), and  zero yaw angle (i.e. $v_2^2+v_4^2=v_1^2+v_3^2$ ), the optimal rotors' speeds can be given by:
 	\begin{align} \label{78}
	\hspace{-0.1cm}	 &\textrm{positive change of pitch angle:} \nonumber\\
 &	\begin{cases}
 	{v_2} = 0,{v_1} = {v_3} = \frac{1}{\sqrt{2}}{v_\textrm{max}},{v_4} = {v_{\max }},\,\, &\textrm{if}\,\,\, 0 < t \le {\tau _1},\\
 	{v_4} = 0,{v_1} = {v_3} = \frac{1}{\sqrt{2}}{v_\textrm{max}},{v_2} = {v_{\max }},\,\,\ &\textrm{if}\,\,\,{\tau _1} < t \le {\tau _2},
 	\end{cases}
 	\end{align} 
 	
 	\begin{align} \label{79}
 &	\hspace{-0.1cm}  	\textrm{positive change of roll angle:} \nonumber\\
 	&\begin{cases}
 	{v_1} = 0,{v_2} = {v_4} = \frac{1}{\sqrt{2}}{v_\textrm{max}},{v_3} = {v_{\max }},\,\,\ &\textrm{if}\,\,\,{\tau _2} < t \le {\tau _3},\\
 	{v_3} = 0,{v_2} = {v_4} = \frac{1}{\sqrt{2}}{v_\textrm{max}},{v_1} = {v_{\max }},\,\,\ &\textrm{if}\,\,\,{\tau _3} < t \le {\tau _4},
 	\end{cases}
 	\end{align} 
 	Therefore, in the first Stage, the drone changes its orientation such that it can move towards $D$ in presence of external forces (e.g., gravity and wind). In the second Stage, the drone moves with a maximum acceleration. In Stage 3, the drone's orientation changes to minimize the acceleration towards $D$. In Stage 4, the drone moves with a minimum acceleration. In Stages 5 and 6, the drone's orientation and the rotors' speeds are adjusted to ensure the stability of drone at $D$. Clearly, the drone will be stable when its total force, $A$ given in (\ref{A_Force}), is zero. Hence, we must have $F=|\boldsymbol{F}_\textrm{ext}|$. Using (\ref{Velocit}) with $T_\textrm{tot}=|\boldsymbol{F}_\textrm{ext}|$, the rotors' speeds in the stable stage is: 
 	\begin{equation} \label{vF_stable}
 	v_\textrm{F}=\sqrt{\frac{|\boldsymbol{F}_\textrm{ext}|}{4\rho_1}}. 
 	\end{equation}
 	The rotors' speed in Stages 1-6 are given in  (\ref{Stage1})-(\ref{Stage6}). 
 	
 	In order to find the switching times, we use the dynamic equations of the drone given in (\ref{Velocit}-\ref{pitc}). For instance, in Stage 1, the time needed for a  $\Delta {\psi _{\textrm{p},1}}$ pitch angle change can be obtained using (\ref{Velocit}) and (\ref{pitc}). In this case, given the rotors' speed in (\ref{Stage1}), and the dynamic equations of the drone, we can find $\tau_1$ and $\tau_2$ as: \vspace{-0.2cm} 
 	\begin{align}
 	&{\tau _1} = \frac{1}{{{v_{\max }}}}\sqrt {\frac{{\Delta {\psi _{\textrm{p},1}}{I_y}}}{{l{\rho _1}}}},\,\,\,\tau_2=2\tau_1, \label{Oriane}
 	\end{align}
 	where $\Delta {\psi _{\textrm{p},1}}$ is the change of pitch angle at Stage 1. Likewise, $\tau_3$ and $\tau_4$ can also be determined.
 	
 	In Stage 2, the time needed for moving within a $d_\textrm{s2}$ distance is given by:
 	\begin{equation}
 	{t_{s2}} = \sqrt {\frac{{2{d_{s2}}A_{s2}}}{{{m_D}}}},
 	\end{equation}
 	where $A_{s2}$ is the total force towards the drone's destination at Stage 2 which can be determined using (\ref{A_Force}). Subsequently, we can find the switching time by $\tau_5=\tau_4+t_{s2}$.
 	
 	The switching times in Stages 3-5 can be determined by adopting the similar approach used in Stages 1 and 2. Note that, $\tau_{14}$ represents the total control time the drone, which can be determined based on (\ref{bigT}) and (\ref{Oriane}) as follows: \vspace{-0.02cm} 
 	\begin{equation}
 	T_{I,D}=\tau_{14}=\sqrt{2d_{D}\Big(\frac{m_D}{A_{s2}}-\frac{m_D}{A_{s4}}\Big)}+T^O,
 	\end{equation}
 	where $A_{s4}$ is the total force on the drone as Stage 4. $T^O$ is the total control time needed for the orientation changes in Stages 1,3, and 5, given by:
 	\begin{align}
 	T^O=&\frac{2}{{{v_{\max }}}}\Big[\sqrt {\frac{{\Delta {\psi _{\textrm{p},1}}{I_y}}}{{l{\rho _1}}}}+\sqrt {\frac{{\Delta {\psi _{\textrm{r},1}}{I_x}}}{{l{\rho _1}}}}+\sqrt {\frac{{\Delta {\psi _{\textrm{p},3}}{I_y}}}{{l{\rho _1}}}}\nonumber\\
 	&+\sqrt {\frac{{\Delta {\psi _{\textrm{r},3}}{I_x}}}{{l{\rho _1}}}}+\sqrt {\frac{{\Delta {\psi _{\textrm{p},5}}{I_y}}}{{l{\rho _1}}}}+\sqrt {\frac{{\Delta {\psi _{\textrm{r},5}}{I_x}}}{{l{\rho _1}}}}\Big],
 	\end{align}
 	where $\Delta\psi _{\textrm{p},i}$, $\Delta\psi _{\textrm{r},i}$ are the pitch and roll changes in Stage $i$.
 	This completes the proof. \vspace{-0.1cm}

\def\baselinestretch{1.05}
\bibliographystyle{IEEEtran}

\bibliography{references}
\end{document}